\documentclass{lmcs} 
\pdfoutput=1

\usepackage{lastpage}
\renewcommand{\dOi}{14(3:21)2018}
\lmcsheading{}{1--\pageref{LastPage}}{}{}%
{Sep.~12,~2017}{Sep.~24,~2018}{}

\keywords{first order logic, logical theory, graph order, subgraph, graph minor, induced subgraph, recursive, computable, poset, partial order}

\usepackage{hyperref}
\usepackage{amsfonts, amsmath, amssymb,amsthm}
\usepackage{color}
\usepackage{mathrsfs}
\usepackage{appendix}
\usepackage{graphicx}
\usepackage{tikz}
\usetikzlibrary{shapes,positioning,fit,calc,decorations.text,decorations.pathmorphing,backgrounds,arrows.meta}

\definecolor{c01}{RGB}{225,128,0} 
\definecolor{c02}{RGB}{60,179,113}
\definecolor{c03}{RGB}{220,20,60}
\definecolor{c04}{RGB}{205,92,92}

\newcommand {\ZEROSEVENTWOCONEs }{\begin{flalign*}
 	2C1s(x,i,j) :=& |x|=i+j+2 \; \wedge \; C_{i \rightarrow 1} \leq_s x \; \wedge \; C_{j \rightarrow 1} \leq_s x \\
 	&  \wedge \; \forall z \; maximalComp(x,z) \supset \mathcal{C}_{\rightarrow 1}(z)
 \end{flalign*}}

 \newcommand{\ZEROTWOpsiUNDERSCOREopres}{\[ \psi_{opres}(x,y) :=  \tilde{\mathcal{G}}(x) \; \wedge \; constructFromCycles(x,y)\]}

 \newcommand{\ZEROTHREEpsiUNDERSCOREedgeOP}{\begin{flalign*}
	\psi_{edgeOP}(x,i,j) :=& \exists y \; x \in \tilde{y}\; \wedge \; \exists m \; (|x|=m^2 +m(m+1)/2 +3m) \; \wedge \\
	 	& CP_4C(m+i+2,m+j+2) \leq_s x
\end{flalign*}}

\newcommand{\ZEROEIGHTaddVert}{\[ addVert(x,y) := conn(x) \; \wedge \; conn(y) \; \wedge \; \exists	z \; y \lessdot_{sv} z \; \wedge \; z \lessdot_{se} x \]}

\newcommand{\ZERONINEallCyclesPRIME}{\[ allCycles'(x,n) := \forall m \; (n+3 \leq m \leq 2n+2) \; \supset \; C_m \leq_s x\]}

\newcommand{\ONEONEbicycle}{\[ bicycle(x) := conn(x) \; \wedge \; \exists y,i,j \; soc2(y,i,j) \; \wedge \; y \lessdot_{se} x \]}

\newcommand{\ONETHREECrightarrowONE}{\[ \mathcal{C}_{\rightarrow 1} (x) := \exists y \; \mathcal{C}(y) \; \wedge \; addVert(x,y)\]}

\newcommand{\ONEFOURCrightarrowTWO}{\[ \mathcal{C}_{\rightarrow 2}(x) := \exists y \; \mathcal{C}_{\rightarrow 1}(y) \; \wedge \; addVert(x,y) \; \wedge \; P_{|x|} \leq_s x \; \wedge \; double3star \nleq_s x\]}

\newcommand{\ONEFIVEcardCond}{\[ cardCond(x,y) :=  |y|= |x|^2 + |x|(|x|+1)/2 + 3|x|\]}

\newcommand{\ONESIXcomp}{\[ comp(y,x) := conn(x) \; \wedge \; \exists z \; x \cup z = y\]}

\newcommand{\ONESEVENcompUnionInZ}{\begin{flalign*}
	compUnionInZ(z,x,y,z')  := & ((z' \leq_s x \vee z' \leq_s y) \supset \\
			& maxCopies(x,z')+maxCopies(y,z')=maxCopies(z,z'))
\end{flalign*}}

\newcommand{\ONEEIGHTcompZInUnion}{\begin{flalign*} compZInUnion(z,x,y,z') := & (z' \leq_s z \supset \\ & maxCopies(x,z')+maxCopies(y,z')=maxCopies(z,z')) \end{flalign*}}

\newcommand{\TWOZEROcopies}{\[ maxCopies(x,y,n) := \exists z \; mult(z,y,n) \; \wedge \; z \leq_s x \; \wedge \; \exists z' \; mult(z,y,n+1) \; \wedge \; \neg z' \leq_s x \qedhere \]}

\newcommand{\TWOONEconstructFromCycles}{\[ constructFromCycles(y,x) :=  \exists z \; z = x \cup \bigcup_{i=1}^{|x|} C_{|x|+i+2} \; \wedge \; |z|=|y| \; \wedge \;  ||z||+|x|=||y|| \qedhere\]}

\newcommand{\TWOTWOcountComps}{\begin{flalign*}
	countComps(x,n) := & \exists y \; extendToCliques(y,x) \; \wedge \; \exists m_1,m_2 \in \mathcal{N} \; \wedge \\ 
 	& makeSequenceFromUOC(y,m_1) \; \wedge \; \psi^{trans}_{sequenceConvert}(m_1,m_2) \\
	& \wedge \; \psi^{trans}_{sequenceSum}(m_2,n)
 \tag*{\qedhere}
\end{flalign*}}

\newcommand{\TWOTHREEcountEdges}{\begin{flalign*}
	countEdges(x,n) := &\exists y \; conn(y) \; \wedge \;  x \leq_s y \; \forall z \; conn(z) \; \wedge \; x \leq_s z \supset \neg (z <_s y) \; \wedge  \\
	& \exists m_1,m_2 \in \mathcal{N} \; countEdgesConn(y,m_1) \; \wedge \; countComps(x,m_2) \\
	&\wedge \; n =m_1 -m_2+1
 \tag*{\qedhere}
\end{flalign*}}

\newcommand{\TWOFOURcountEdgesConn}{\[ countEdgesConn(x,n) :=  \exists y \; countEdgesGadget(y,x) \; \wedge \; {|x| \choose 2} - n +1 = |y| \qedhere\]}

\newcommand{\TWOFIVEcountEdgesGadget}{\begin{flalign*}
	countEdgesGadget(y,x) := & countEdgesGadget'(y,x) \\
			& \wedge \; \forall z \; countEdgesGadget'(z,x) \supset \neg(z <_s y)
\end{flalign*}}

\newcommand{\TWOSIXcountEdgesGadgetPRIME}{\begin{flalign*}
			countEdgesGadget'(y,x) := & conn(x) \; \wedge \; comp(y,x) \; \wedge \; \forall z \; comp(y,z) \supset \\
			& \quad  [x \leq_s z \; \wedge \; |x|=|z|  \; \wedge \\
			& \quad (\mathcal{K}(z) \vee (\exists z' \; (comp(y,z') \; \wedge \; z \lessdot_{se} z' )))]
\end{flalign*}}

\newcommand{\TWOSEVENCPUNDERSCOREFOURC}{\begin{flalign*}
	CP_4C(x,i,j) := & conn(x) \; \wedge \; \mathcal{N}(i) \; \wedge \; \mathcal{N}(i) \; \wedge \; 3 < i <j \\
	& C_{i \rightarrow 1} \cup C_{j \rightarrow 1} \lessdot_{se} x \; \wedge \;  \forall z \; z \leq_s x \\
	& \quad  \wedge \; (addVert(z,C_{i \rightarrow 1}) \vee addVert(z,C_{j \rightarrow 1} )) \supset \mathcal{C}_{\rightarrow 2}(z)
\end{flalign*}}

\newcommand{\TWOEIGHTcsum}{\begin{flalign*}
	csum(x,n) := & \mathcal{N} (n) \; \wedge \; \forall z \; maximalComp(x,z) \; \supset \; \mathcal{C}(z) \; \wedge \\
	&  cardCond(x,n) \;  \wedge \; allCycles'(x,n)
\end{flalign*}}

\newcommand{\TWONINEdisjointUnion}{\begin{flalign*}
	disjointUnion(z,x,y) := &\forall z' \; conn(z') \; \wedge  compUnionInZ(z,x,y,z') \; \wedge \\
	& compZInUnion(z,x,y,z')
\end{flalign*}}

\newcommand{\THREEONEedgeMaximal}{\[ edgeMaximal(x,y) := \forall z \; x \lessdot_{se} z \supset \neg maximumComp(z,y) \]}

\newcommand{\THREETWOextendToCliques}{\begin{flalign*}
	extendToCliques(y,x) := & unionOfCliques(y) \; \wedge \; x \leq_s y \; \wedge \; \\
	& \forall z \; (unionOfCliques(z) \; \wedge \; x \leq_s z) \supset \neg(z <_s y)
\end{flalign*}}

\newcommand{\THREEFIVEmakeSequenceFromUOC}{\begin{flalign*}
	makeSequenceFromUOC(x,n) := &unionOfCliques(x) \; \wedge \; \forall i \in \mathcal{N} \\
	& (\psi^{trans}_{sequence}(n,i,j) \iff maxCopies(K_i,x)=j)
\end{flalign*}}

\newcommand{\THREESIXmaximalComp}{\[ maximalComp(y,x) := conn(x) \; \wedge \; x \leq_s y \; \wedge \; \forall z \; conn(z) \; \wedge \; z \leq_s y \supset \neg (x <_s z) \]}

\newcommand{\THREESEVENmaximumComp}{\[ maximumComp(x,y) := maximalComp(x,y) \; \wedge \; \forall z \; maximalComp(x,z) \supset z \leq_s y\]}

\newcommand{\THREEEIGHTmult}{\begin{flalign*}
	 mult(x,y) := & zero(x) \; \vee \; (conn(y) \; \wedge \; maximumComp(x,y) \; \wedge \\
	 &	uniqueCompCard(x,y)\; \wedge \; edgeMaximal(x,y) )
\end{flalign*}}

\newcommand{\FOURFIVEpointedCycleSum}{\begin{flalign*}
	pointedCycleSum(x,i,j) := & conn(x) \; \wedge \; \exists z \; C_i \cup C_j \lessdot_{sv} z \\
	& \quad \wedge \; z \lessdot_{se}^2 x \; \wedge \; \forall w \; bicycle(w) \supset w \nleq x
\end{flalign*}}

\newcommand{\FOUREIGHTsameSize}{\[ sameSize(x,y) := \exists n \; countEdges(x,n) \; \wedge \; countEdges(y,n) \]}

\newcommand{\FIVETHREEsocTWO}{\begin{flalign*}
			soc2(x,i,j) :=& \mathcal{N}(i) \; \wedge \; \mathcal{N}(j) \; \wedge \; soc(x)  \; \wedge \; |x|=i+j  \\
			& \wedge \; C_i \leq_s x \; \wedge \; C_j \leq_s x \; \wedge \\
			 & \forall y \; ( \mathcal{C}(y) \; \wedge \; y \leq_s x) \supset (y=C_i \vee y=C_j)  
			\end{flalign*}}

\newcommand{\FIVEFOURunionOfCliques}{\[ unionOfCliques(x) := \forall x'\; comp(x,x') \supset \mathcal{K}(x') \]}

\newcommand{\FIVEFIVEuniqueCompCard}{\begin{flalign*}
	uniqueCompCard(x,y) := & \forall z \; x \lessdot_{se} z \supset (\forall z' \; maximalComp(z,z') \supset \\
		& \quad (|z'| = |y| \; \vee \; |z'|=2|y|))
\end{flalign*}}

\usepackage{complexity}

\theoremstyle{plain} 

\begin{document}

\title[Defining Recursive Predicates in Graph Orders]{Defining Recursive Predicates in Graph Orders}
\titlecomment{{\lsuper*} This is an extended version of \cite{thinniyam2017definability} which appeared in the proceedings of ICLA 2017.}

\author[R.S.~Thinniyam]{Ramanathan S. Thinniyam}  
\address{The Institute of Mathematical Sciences, 4th Cross Street, CIT Campus, Tharamani, Chennai, Tamil Nadu 600113, India} 
\address{Homi Bhabha National Institute, Mumbai, India}
\email{thinniyam@imsc.res.in}  







 \begin{abstract}
  \noindent We study the first order theory of \textit{structures over graphs} i.e. structures of the form ($\mathcal{G},\tau$) where $\mathcal{G}$ is the set of all (isomorphism types of) finite undirected graphs and $\tau$ some vocabulary. We define the notion of a \textit{recursive predicate over graphs} using Turing Machine recognizable string encodings of graphs. We also define the notion of an \textit{arithmetical relation over graphs}(borrowed from \cite{kudinov2009gandy}) using a total order $\leq_t$ on the set $\mathcal{G}$ such that ($\mathcal{G},\leq_t$) is isomorphic to ($\mathbb{N},\leq$). 

  We introduce the notion of a \textit{capable} structure over graphs, which is an arithmetical structure satisfying the conditions : (1) definability of arithmetic, (2) definability of cardinality of a graph, and (3) definability of two particular graph predicates related to vertex labellings of graphs. We then show any capable structure can define every arithmetical predicate over graphs. As a corollary, any capable structure also defines every recursive graph relation. We identify capable structures which are expansions of graph orders, which are structures of the form ($\mathcal{G},\leq$) where $\leq$ is interpreted as a partial order. We show that the subgraph order i.e. ($\mathcal{G},\leq_s$), induced subgraph order with one constant $P_3$ i.e. ($\mathcal{G},\leq_i,P_3$) and an expansion of the minor order for counting edges i.e. ($\mathcal{G},\leq_m,sameSize(x,y)$) are capable structures. In the course of the proof, we show the definability of several natural graph theoretic predicates in the subgraph order which may be of independent interest. We discuss the implications of our results and connections to Descriptive Complexity.
 \end{abstract}

\maketitle

\section{Introduction}
\label{sec:intro}
	Graphs are combinatorial objects used to model a variety of real world problems \cite{foulds2012graph}. There are many variants such as directed graphs, undirected graphs, weighted graphs etc. which may used depending on the application. In this paper by graph we will mean simple, finite, unlabelled graph. Computational questions about finite graphs have assumed importance in computer science; for instance, in computational complexity \cite{garey1976some}. One way to undertake a logical study of graphs is the way taken by Finite Model Theory \cite{gradel2007finite} and the associated field of Descriptive Complexity \cite{immerman2012descriptive}, which associates computational complexity classes with logics. These logics have a vocabulary which extends the edge relation. 

	In this paper however, we take up the logical study of graphs in a different way. Consider a structure whose domain is the set of all finite isomorphism types of undirected, simple graphs,  denoted $\mathcal{G}$. We will call a structure ($\mathcal{G},\tau$) with vocabulary $\tau=(\mathscr{F},\mathscr{R},\mathscr{C})$ of functions, relations and constants a \textit{structure over graphs}. Note that a variable interpreted in such a structure is assigned an isomorphism type of finite graphs.

	Consider the following well known theorem in graph theory:\\
	``A graph is bipartite if and only if it does not contain any odd cycle.''\\
	Suppose we are given the predicate $bipartite(x)$ if and only if $x$ is a bipartite graph and the predicate $oddCycle(x)$ if and only if $x$ is a cycle of odd order. We can write the theorem in the first order theory of the object ($\mathcal{G},\leq_s,bipartite,oddCycle$) :
	\[ \forall x \; [bipartite(x) \iff (\forall y \; oddCycle(y) \supset y \nleq_s x )] \]
	Such a formulation is not possible in Finite Model Theory where the paradigm is ``finite graph as model'' i.e. we cannot write statements which relate two (or more) graphs.

	  There are many candidates for the vocabulary $\tau$ since there are many natural graph theoretic predicates of interest. For example we could have a unary predicate $\chi_n(x)$ which is true of a graph $g$ if and only if $g$ has chromatic number $n$. 
    We consider the simple vocabulary consisting of a single order symbol $\leq$ in the context of structures over graphs as our starting point for the logical study of structures over graphs. 

	Many natural relations such as subgraph, induced subgraph and minor form partial orders on the set $\mathcal{G}$ (see Figure \ref{fig:indSubOrd} for the induced subgraph order). We take up the study of the first order theory of such \textit{graph orders}. For instance, by the subgraph order we mean the object  ($\mathcal{G},\leq_s$) with domain $\mathcal{G}$ and a single binary order symbol $\leq_s$ whose interpretation is fixed as follows : for two graphs $g_1,g_2 \in \mathcal{G}$, $g_1 \leq_s g_2$ if and only if $g_1$ is a subgraph of $g_2$. Similarly $\leq_m$ denotes the minor order and $\leq_i$ the induced subgraph order . Note that the edge relation of any particular graph is not known to us (see Figure \ref{fig:indSubAsFOstr}).
  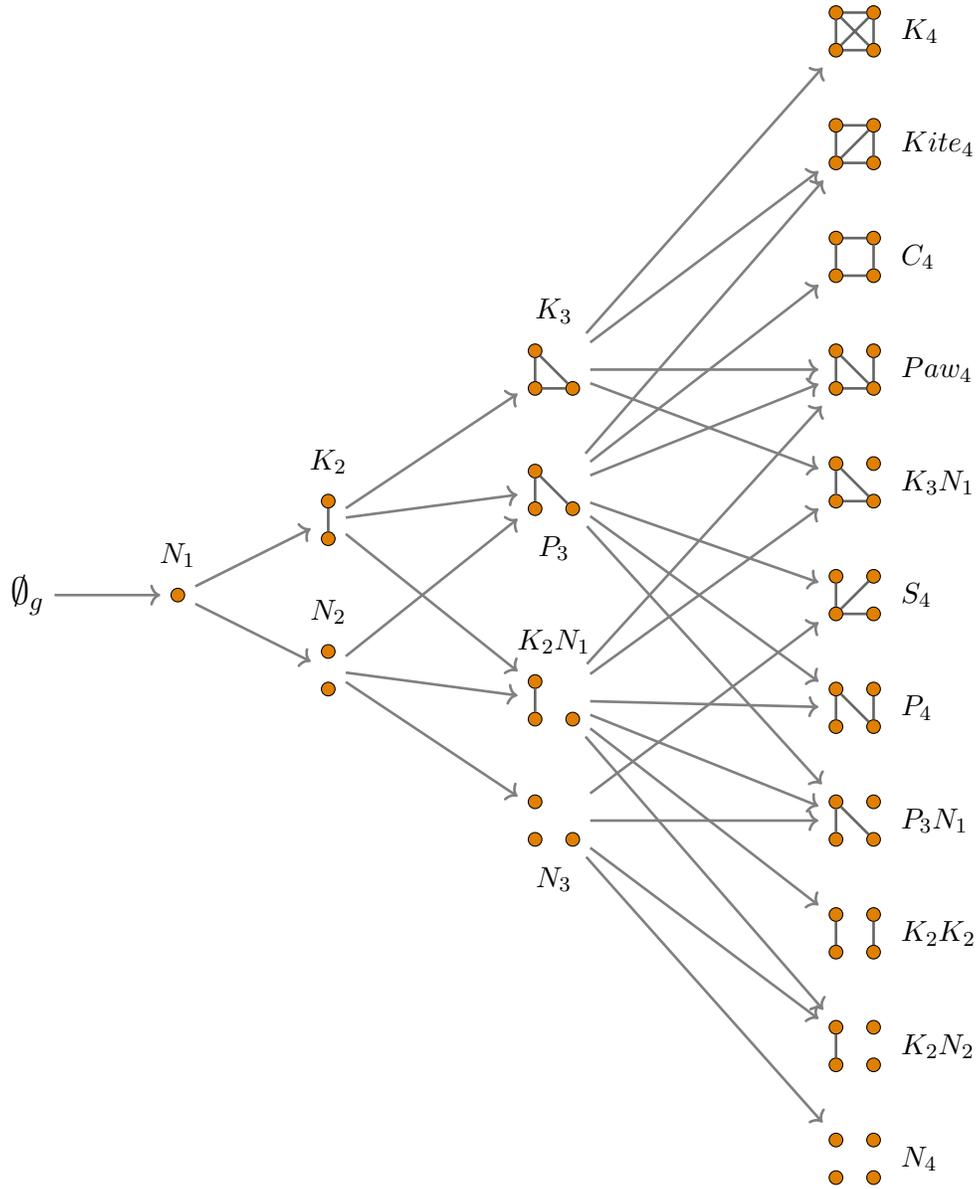
\begin{figure}
 \begin{tikzpicture}
 \def \r {0.7cm}
 \def \s {0.5}
 \def \d {2cm}
 \def \et {1}
 \def \el {0.4}
 \def \ct {1}

 \tikzstyle{node1} = [draw, circle, fill=c01, scale=\s]
 \tikzstyle{node2} = [ ]

 \tikzstyle{line1}=[draw,line width=\et,color=black!60]
 \tikzstyle{line2}=[draw, ->, line width=\ct, color=black!50] 
  \node[node2] (empty) at (0,0) {\Large $\emptyset_g$};
  \node[node2,label=above:$N_1$] (N1) at (\d,0) {
   \begin{tikzpicture}
     \node[node1] at (0,0) {};
   \end{tikzpicture}
  };
  \node[node2,label=above:$K_2$] (K2) at (2*\d,0.5*\d) {
    \begin{tikzpicture}
     \node[node1] (v1) at (0,0) {};
     \node[node1] (v2) at (0,0.5) {};
     \draw[line1] (v1) -- (v2);
   \end{tikzpicture}
  };
  \node[node2,label=above:$N_2$] (N2) at (2*\d,-0.5*\d) {
    \begin{tikzpicture}
     \node[node1] (v1) at (0,0) {};
     \node[node1] (v2) at (0,0.5) {};
   \end{tikzpicture}
  };

  \node[node2,label=above:$K_3$] (K3) at (3.5*\d,1.5*\d) {
    \begin{tikzpicture}
     \node[node1] (v1) at (0,0) {};
     \node[node1] (v2) at (0,0.5) {};
     \node[node1] (v3) at (0.5,0){};
     \draw[line1] (v1)--(v2)--(v3)--(v1);
    \end{tikzpicture}
  };
  \node[node2,label=below:$P_3$] (P3) at (3.5*\d,0.7*\d) {
    \begin{tikzpicture}
     \node[node1] (v1) at (0,0) {};
     \node[node1] (v2) at (0,0.5) {};
     \node[node1] (v3) at (0.5,0){};
     \draw[line1] (v1)--(v2)--(v3);
    \end{tikzpicture}
  };
  \node[node2,label=above: $K_2N_1$] (K2N1) at (3.5*\d,-0.7*\d) {
    \begin{tikzpicture}
     \node[node1] (v1) at (0,0) {};
     \node[node1] (v2) at (0,0.5) {};
     \node[node1] (v3) at (0.5,0){};
     \draw[line1] (v1)--(v2);
    \end{tikzpicture}
  };
  \node[node2,label=below:$N_3$] (N3) at (3.5*\d,-1.5*\d) {
    \begin{tikzpicture}
     \node[node1] (v1) at (0,0) {};
     \node[node1] (v2) at (0,0.5) {};
     \node[node1] (v3) at (0.5,0){};
    \end{tikzpicture}
  };

  \node[node2,label=right:$K_4$] (K4) at (5.5*\d,5*0.75*\d) {
    \begin{tikzpicture}
     \node[node1] (v1) at (0,0) {};
     \node[node1] (v2) at (0,0.5) {};
     \node[node1] (v3) at (0.5,0){};
     \node[node1] (v4) at (0.5,0.5) {};
     \draw[line1] (v1)--(v2)--(v3)--(v4)--(v1)--(v3);
     \draw[line1] (v2)--(v4); 
    \end{tikzpicture}
  };
  \node[node2,label=right:$Kite_4$] (Kite4) at (5.5*\d,4*0.75*\d) {
    \begin{tikzpicture}
     \node[node1] (v1) at (0,0) {};
     \node[node1] (v2) at (0,0.5) {};
     \node[node1] (v3) at (0.5,0){};
     \node[node1] (v4) at (0.5,0.5) {};
     \draw[line1] (v1)--(v2)--(v4)--(v3)--(v1)--(v4);
    \end{tikzpicture}
  };
  \node[node2,label=right:$C_4$] (C4) at (5.5*\d,3*0.75*\d) {
    \begin{tikzpicture}
     \node[node1] (v1) at (0,0) {};
     \node[node1] (v2) at (0,0.5) {};
     \node[node1] (v3) at (0.5,0){};
     \node[node1] (v4) at (0.5,0.5) {};
     \draw[line1] (v1)--(v2)--(v4)--(v3)--(v1);
    \end{tikzpicture}
  };
  \node[node2,label=right:$Paw_4$] (Paw4) at (5.5*\d,2*0.75*\d) {
    \begin{tikzpicture}
     \node[node1] (v1) at (0,0) {};
     \node[node1] (v2) at (0,0.5) {};
     \node[node1] (v3) at (0.5,0){};
     \node[node1] (v4) at (0.5,0.5) {};
     \draw[line1] (v4)--(v3)--(v2)--(v1)--(v3); 
    \end{tikzpicture}
  };
  \node[node2,label=right:$K_3N_1$] (K3N1) at (5.5*\d,0.75*\d) {
    \begin{tikzpicture}
     \node[node1] (v1) at (0,0) {};
     \node[node1] (v2) at (0,0.5) {};
     \node[node1] (v3) at (0.5,0){};
     \node[node1] (v4) at (0.5,0.5) {};
     \draw[line1] (v1)--(v2)--(v3)--(v1); 
    \end{tikzpicture}
  };
  \node[node2,label=right:$S_4$] (S4) at (5.5*\d,0) {
    \begin{tikzpicture}
     \node[node1] (v1) at (0,0) {};
     \node[node1] (v2) at (0,0.5) {};
     \node[node1] (v3) at (0.5,0){};
     \node[node1] (v4) at (0.5,0.5) {};
     \draw[line1] (v3)--(v1)--(v2);
     \draw[line1] (v4)--(v1);
    \end{tikzpicture}
  };
  \node[node2,label=right:$P_4$] (P4) at (5.5*\d,-1*0.75*\d) {
    \begin{tikzpicture}
     \node[node1] (v1) at (0,0) {};
     \node[node1] (v2) at (0,0.5) {};
     \node[node1] (v3) at (0.5,0){};
     \node[node1] (v4) at (0.5,0.5) {};
     \draw[line1] (v1)--(v2)--(v3)--(v4); 
    \end{tikzpicture}
  };
  \node[node2,label=right:$P_3N_1$] (P3N1) at (5.5*\d,-2*0.75*\d) {
    \begin{tikzpicture}
     \node[node1] (v1) at (0,0) {};
     \node[node1] (v2) at (0,0.5) {};
     \node[node1] (v3) at (0.5,0){};
     \node[node1] (v4) at (0.5,0.5) {};
     \draw[line1] (v1)--(v2)--(v3); 
    \end{tikzpicture}
  };
  \node[node2,label=right:$K_2K_2$] (K2K2) at (5.5*\d,-3*0.75*\d) {
    \begin{tikzpicture}
     \node[node1] (v1) at (0,0) {};
     \node[node1] (v2) at (0,0.5) {};
     \node[node1] (v3) at (0.5,0){};
     \node[node1] (v4) at (0.5,0.5) {};
     \draw[line1] (v1)--(v2);
     \draw[line1] (v3)--(v4); 
    \end{tikzpicture}
  };
  \node[node2,label=right:$K_2N_2$] (K2N2) at (5.5*\d,-4*0.75*\d) {
    \begin{tikzpicture}
     \node[node1] (v1) at (0,0) {};
     \node[node1] (v2) at (0,0.5) {};
     \node[node1] (v3) at (0.5,0){};
     \node[node1] (v4) at (0.5,0.5) {};
     \draw[line1] (v1)--(v2); 
    \end{tikzpicture}
  };
  \node[node2,label=right:$N_4$] (N4) at (5.5*\d,-5*0.75*\d) {
    \begin{tikzpicture}
     \node[node1] (v1) at (0,0) {};
     \node[node1] (v2) at (0,0.5) {};
     \node[node1] (v3) at (0.5,0){};
     \node[node1] (v4) at (0.5,0.5) {};
    \end{tikzpicture}
  };

  \draw[line2] (empty)--(N1);
  \draw[line2] (N1)--(K2);
  \draw[line2] (N1)--(N2);

  \draw[line2] (K2)--(K3);
  \draw[line2] (K2)--(P3);
  \draw[line2] (K2)--(K2N1);

  \draw[line2] (N2)--(P3);
  \draw[line2] (N2)--(K2N1);
  \draw[line2] (N2)--(N3);

  \draw[line2] (K3)--(K4);
  \draw[line2] (K3)--(Kite4);
  \draw[line2] (K3)--(Paw4);
  \draw[line2] (K3)--(K3N1);

  \draw[line2] (P3)--(Kite4);
  \draw[line2] (P3)--(C4);
  \draw[line2] (P3)--(Paw4);
  \draw[line2] (P3)--(S4);
  \draw[line2] (P3)--(P4);
  \draw[line2] (P3)--(P3N1);

  \draw[line2] (K2N1)--(Paw4);
  \draw[line2] (K2N1)--(K3N1);
  \draw[line2] (K2N1)--(P4);
  \draw[line2] (K2N1)--(P3N1);
  \draw[line2] (K2N1)--(K2K2);
  \draw[line2] (K2N1)--(K2N2);

  \draw[line2] (N3)--(S4);
  \draw[line2] (N3)--(P3N1);
  \draw[line2] (N3)--(K2N2);
  \draw[line2] (N3)--(N4);

 \end{tikzpicture}
  \caption{The first few levels of the induced subgraph order $\leq_i$. Note the symmetry arising from the automorphism $f(g)= g^c$.}
    \label{fig:indSubOrd}
  \end{figure}
  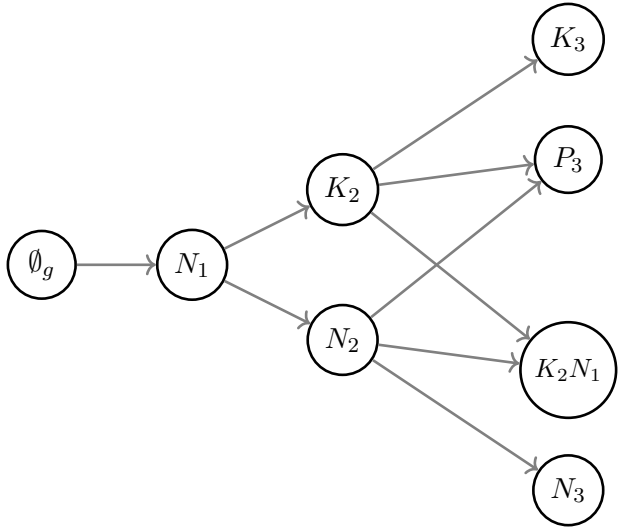
\begin{figure}
 \begin{tikzpicture}
 \def \r {0.7cm}
 \def \s {0.5}
 \def \d {2cm}
 \def \et {1}
 \def \el {0.4}
 \def \ct {1}

 \tikzstyle{node1} = [draw, circle, fill=c01, scale=\s]
 \tikzstyle{node2} = [draw, circle, line width=1, scale=2*\s]

 \tikzstyle{line1}=[draw,line width=\et,color=black!60]
 \tikzstyle{line2}=[draw, ->, line width=\ct,color=black!50] 
 %
  \node[node2] (empty) at (0,0) {$\emptyset_g$};
  \node[node2] (N1) at (\d,0) {$N_1$};

  \node[node2] (K2) at (2*\d,0.5*\d) {$K_2$};
  \node[node2] (N2) at (2*\d,-0.5*\d) {$N_2$};

  \node[node2] (K3) at (3.5*\d,1.5*\d) {$K_3$};
  \node[node2] (P3) at (3.5*\d,0.7*\d) {$P_3$};
  \node[node2] (K2N1) at (3.5*\d,-0.7*\d) {\small $K_2N_1$};
  \node[node2] (N3) at (3.5*\d,-1.5*\d) {$N_3$};
  \draw[line2] (empty)--(N1);
  \draw[line2] (N1)--(K2);
  \draw[line2] (N1)--(N2);

  \draw[line2] (K2)--(K3);
  \draw[line2] (K2)--(P3);
  \draw[line2] (K2)--(K2N1);

  \draw[line2] (N2)--(P3);
  \draw[line2] (N2)--(K2N1);
  \draw[line2] (N2)--(N3);

 \end{tikzpicture}
  \caption{Initial layers of the object ($\mathcal{G},\leq_i$). Note that the edge relation i.e. the ``internal structure'' is not available to us directly, nor are the names indicated inside the nodes. The arrows indicate the upper cover relation.}
    \label{fig:indSubAsFOstr}
  \end{figure}

	The domain $\mathcal{G}$ has however been studied in a different, non-logical setting : it appears in the graph homomorphism literature \cite{hell2004graphs}. The object ($\mathcal{G},\leq_h$) where $g_1 \leq_h g_2$ if and only if there exists a homomorphism from $g_1$ to $g_2$ is a preorder and not a partial order. Appropriate quotienting gives us the poset ($\mathscr{C},\leq_h$) where $\mathscr{C}$ is the set of \textit{cores}, which are minimum elements under the $\leq_h$ order inside a particular equivalence class. However much of the literature on graph homomorphism concentrates on \textit{homomorphism densities} \cite{lovasz2012random,lovasz2008graph,razborov20103}. Hatami and Norine in \cite{hatami2011undecidability},  start the paper by saying ``Many fundamental theorems in extremal graph theory can be expressed as algebraic inequalities between subgraph densities...for dense graphs it is possible to replace subgraph densities with homomorphism densities'', thus motivating this line of research. The above paper proves the undecidability of linear inequalities over graph homomorphism densities, showing the difficulty of general problems even in this restricted language. 

	Our work can be thought of as extending the work of Je\v{z}ek and McKenzie  who study the substructure orderings over finite lattices \cite{jevzek2009bdefinability}, semilattices \cite{jevzek2009definability}, distributive lattices \cite{jevzek2009adefinability} and posets \cite{jevzek2010definability}. This line of work was extended to the induced subgraph order (equivalently, the substructure ordering over graphs) by Wires \cite{wires2016definability} and to directed graphs by Kunos \cite{kunos2015definability}.

 One of the objectives in these studies is to understand universal classes of objects, not necessarily finite. For instance, let $\mathcal{K}$ be a class of semi-lattices described by some set of universal sentences. Ordering all such $\mathcal{K}$ by inclusion gives the lattice $\mathcal{U}$. Let $\mathcal{S}$ be the set of all finite semi-lattices. Then the map sending $\mathcal{K}$ to $\mathcal{K} \cap \mathcal{S}$ is an isomorphism between $\mathcal{U}$ and the set of all order-ideals of $\mathcal{S}$. Je\v{z}ek and McKenzie observe that the definability in $\mathcal{U}$ of the set of $\mathcal{K}$ which contain only finitely many elements can be reduced to showing the definability in $\mathcal{S}$ of every domain element as a constant. They proceed to show that this latter fact holds. 

	 A key technical construction introduced in \cite{jevzek2010definability} in connection with definability of constants in the substructure orders is that of the objects called``o-presentations'', which occupy centerstage in our work. These objects were defined for graphs by Wires to show that the induced subgraph order has a unique non-trivial automorphism, which maps every graph to its complement. He also shows that the set of predicates definable in ($\mathcal{G},\leq_i,P_3$) (where $P_3$ is a constant for the path on three vertices) is exactly the isomorphism invariant predicates definable in the first order theory of a simple expansion $\mathcal{CG}'$ of the small category of graphs.

	Our work can also be thought of as extending the study of structures of order over words and trees to graphs. While structures with the domain set $\Sigma^*$ of finite words over an alphabet $\Sigma$ have been studied for some time (see Quine \cite{quine1946concatenation} on ($\Sigma^*,.$) where $.$ is a binary function symbol for concatenation), the study of order theories of combinatorial objects is more recent and arose as part of the term rewriting literature \cite{venkataraman1987decidability,dershowitz1979orderings,comon1994ordering}. Treinen and Comon \cite{comon1997first} studied the lexicographic path ordering on words. Other orders such as subword, infix etc  were studied by Kuske \cite{kuske2006theories}. The focus of Kuske's work was on showing the undecidability of these word orders, with an emphasis on the syntactic fragments. More recently, fragments of the subword order have been studied. The $FO^2$ fragment of the subword order was shown to be NP-complete by Schnoebelen and Karandikar \cite{karandikar2015decidability}. The surprising result that the existential fragment extended by constants is undecidable was recently shown by Halfon etal \cite{halfon2017decidability}. 

	The current paper is a continuation of our work \cite{ramanujam2016definability} which shows the definability of arithmetic in the subgraph and minor orders; the analogous result for the induced subgraph order having already been established by Wires \cite{wires2016definability}.  This paper focuses on the computational aspect of graph orders which was not studied in \cite{ramanujam2016definability}. We introduce an encoding of graphs as strings (the actual encoding turns out to be unimportant as we discuss in Section \ref{sec:discussion}). Then we define the notion of a \textit{recursive relation over graphs} by looking at the relations over such string encodings (of appropriate arity depending on the arity of the predicate) recognised by Turing Machines. We also define the notion of an \textit{arithmetical relation over graphs} borrowed from the Theory of Numberings (see \cite{kudinov2009gandy} for a more general definition), which can be thought of as the analog in the setting of graphs to a relation over numbers which belongs to the arithmetical hierarchy. 

	In Section \ref{sec:condResultGenStr} we introduce the notion of a \textit{capable} structure over graphs which is an arithmetical structure over graphs satisfying the conditions: ($C1$) definability of arithmetic, ($C2$) definability of two graph predicates related to o-presentations and ($C3$) definability of the cardinality of a graph. We then show that every arithmetical relation over graphs is definable in a capable structure, in other words, every capable structure has the \textit{maximal definability property} (m.d.p.). The m.d.p. has been established in literature for structures such as the subword order \cite{kudinov2010definability}, but has not been studied extensively for structures over graphs. As a corollary, we obtain the fact that every recursive relation over graphs is definable in any capable structure. 

	 Our main result is to establish that the following structures are capable: ($\mathcal{G},\leq_i,P_3)$ ($P_3$ is a constant symbol for the (isomorphism type of the) path on three vertices, see Figure \ref{fig:indSubOrd}), ($\mathcal{G},\leq_s$) and ($\mathcal{G},\leq_m,sameSize$) in Section \ref{sec:graphOrders}. We note that definability of arithmetic does not automatically imply the main result of this paper. It also requires that the structure of interest be able to access the ``internal structure'' of a graph since the edge relation is not part of the vocabulary(see Figure \ref{fig:indSubAsFOstr}). 

	The fact that graph orders satisfy ($C1$) has already been established (\cite{ramanujam2016definability}, \cite{wires2016definability}).The main technical contribution of this paper is the establishment of condition ($C2$) f{}or the induced subgraph and subgraph orders in Section \ref{sub:induced_subgraph_order}, \ref{sub:subgraph_order}. In the course of doing so, we prove the definability in the subgraph order of graph theoretic predicates of independent interest, such as disjoint union of graphs and number of edges of a graph. As a corollary we are also able to prove the result for an expansion of the minor order in Section \ref{sub:minor_order}.

	 This paper is an extended version of the conference paper \cite{thinniyam2017definability}, which  proves the definability of recursive predicates in the induced subgraph order. The results in this paper can also be seen as a strengthening of the definability results obtained in \cite{ramanujam2016definability}. 

\section{Preliminaries}
 \label{sec:prelim}
 \subsection{Graphs and Graph Orders}
 \begin{defi}[Labelled Graph]
 \label{defi:labelledGraph}
  A (finite) labelled graph $g$ is a structure \\ ($V_g,E_g,L_g$) with
  \begin{enumerate}
     \item finite domain (aka vertex) set $V_g$,
     \item an irreflexive, symmetric binary relation $E_g \subseteq V_g \times V_g$ which is the edge set of the graph, and
     \item a bijective function $L_g:V_g \rightarrow [n]$ where $[n]$ stands for the initial segment $\{1,2,3...,n\}$ of the natural numbers with $n=|V_g|$ i.e. $n$ is the number of vertices in the graph.
   \end{enumerate}
 \end{defi}
 We will write $v_i$ to denote the vertex whose image under $L_g$ is $i$. We will write $v_iv_j$ to denote the edge (if it exists) between $v_i$ and $v_j$. Note that we restrict ourselves to simple graphs i.e. graphs which do not have edges of the form $(v_i,v_i)$ for any $v_i \in V_g$.
 \begin{defi}
  An isomorphism between two labelled graphs $g_1$ and $g_2$ is a bijection $\eta : V_{g_1} \rightarrow V_{g_2}$ such that for any two vertices $v_i,v_j$ of $g_1$, the edge $v_iv_j$ exists if and only if there is an edge between vertices $\eta(v_i),\eta(v_j)$ in $g_2$. 
 \end{defi}
 We say $g_1$ is isomorphic to $g_2$ if there is an isomorphism between them, and write $g_1 \simeq g_2$. The relation $\simeq$ is an equivalence relation on the set of all finite labelled graphs. 
 \begin{defi}[Graph]
  By a graph $g$, we mean an equivalence class under the relation $\simeq$ over the set of all finite labelled graphs. The set of all graphs will be denoted $\mathcal{G}$.
 \end{defi}
 We will write $g =[g']$ to denote that the  graph $g$ is the isomorphism type of the labelled graph $g'$.

 \begin{defi}[Graph Orders]
 Consider the following graph operations:\\
 (O1) Deletion of an edge.\\
 (O2) Deletion of a vertex.\\
 (O3) Contraction of an edge.

 A graph $g_1$ is said to be an induced subgraph of a graph $g_2$ if $g_1$ can be obtained from $g_2$ by  finitely many applications of operation O2. We write $g_1 \leq_i g_2$.

 A graph $g_1$ is said to be a subgraph of a graph $g_2$ if $g_1$ can be obtained from $g_2$ by  finitely many applications of operations O1 and O2. We write $g_1 \leq_s g_2$.

 A graph $g_1$ is said to be a minor of a graph $g_2$ if $g_1$ can be obtained from $g_2$ by  finitely many applications of operations O1, O2 and O3. We write $g_1 \leq_m g_2$.
 \end{defi}
 Note that when a vertex is deleted, so are all edges incident on the vertex. Contraction of an edge $e=uv$ in a graph $g$ to give a graph $g'$ is the deletion of $e$ and identification of both $u$ and $v$ as the same vertex $v'$ in the new graph $g'$. Neither $u$ nor $v$ exist in $g'$ and all edges incident on either $u$ or $v$ are now incident on $v'$. For more on graph minors, see Diestel \cite{diestel2005graph}. 
 \subsection{First Order Structures and Definability}
 \label{sub:prelimsFOstrs}
 For the standard syntax and semantics of first order logic, we refer the reader to Enderton \cite{enderton2001mathematical}. We introduce the notion of a \textit{structure over graphs.}
 \begin{defi}[ Structure over Graphs]
  A structure over graphs is one which has as its domain the set $\mathcal{G}$ and a signature $\tau= (\mathscr{F},\mathscr{R},\mathscr{C})$ of functions $\mathscr{F}$, relations $\mathscr{R}$ and constants $\mathscr{C}$ in $\mathcal{G}$. We will denote the structure by ($\mathcal{G},\tau$).

  In this paper, we will primarily be concerned with  structures over graphs where $\tau$ contains a symbol $\leq$ which will be interpreted as a partial order over $\mathcal{G}$. In particular, we consider the structures
  \begin{enumerate}
     \item ($\mathcal{G},\leq_i,P_3$) i.e. the induced subgraph order with a constant symbol $P_3$ for the path on three vertices.
     \item ($\mathcal{G},\leq_s$) i.e. the subgraph order.
     \item ($\mathcal{G},\leq_m,sameSize$) i.e. the minor order with an additional binary relation $sameSize(x,y)$ which holds if and only if $x$ and $y$ have the same number of edges. 
   \end{enumerate}
 \end{defi}
 The constant $P_3$ is used to break the symmetry of the induced subgraph order which by itself cannot distinguish between a graph and its complement since the map sending a graph to its complement is an automorphism of the order.

 \begin{defi}[Covering Relation of a Poset]
  \label{def:covering_relation}
  Given elements $x,y$ of a poset ($P,\leq$) we define the covering relation $x \lessdot y$ as\\
   $x \lessdot y$ if and only if $x < y$ and there exists no element $z$ of $P$ such that $x < z < y$.
 \end{defi}

 \begin{rem}
  \label{rem:coverSubVsInd}
  The deletion of a vertex of a graph can be thought of as the deletion of some number of edges followed by deletion of an isolated vertex. Thus $g_1 \leq_i g_2$ implies that $g_1 \leq_s g_2$ but the converse need not be true. In addition, $g_1 \lessdot_i g_2$ implies that $g_1 \leq_s g_2$ but not necessarily that $g_1 \lessdot_s g_2$.
 \end{rem}

 \begin{obs}
  \label{obs:defCoverInPoset}
  The covering relation of a poset ($P,\leq$) is first order definable using $\leq$:
  \[ x \lessdot y := x < y \wedge \forall z \neg(x < z < y)\]
 \end{obs}

 As we will see, while the overall proof strategy for proving the definability of recursive predicates in $\leq_s$ and $\leq_i$ remains the same, there are differences at a finer level of granularity. It is easier to define certain predicates in one order as compared to the other and vice versa. We would like to state that certain large cycles occur as induced subgraphs. This can be directly stated using $\leq_i$, but requires us to do more work when using $\leq_s$.

 \begin{defi}[Arithmetic]
  By arithmetic, we mean the first order theory of the structure ($\mathbb{N},\phi_{+},\phi_{\times}$) where $\mathbb{N}$ is the set of all natural numbers and $\phi_{+},\phi_{\times}$ are ternary relations for addition and multiplication respectively.
 \end{defi}
 We will also use variables $x,y,z$ to denote numbers in arithmetical formulae; and lower case letters $i,j,k,l,m,n$ to denote numbers.

 \begin{defi}[Constant Definability]  
  Fix a first order language $\mathcal{L}$. Let $g$ be an element of the domain of an $\mathcal{L}$-structure $\mathcal{A}$. 
 We say that $g$ is definable in $\mathcal{A}$,  if there exists an $\mathcal{L}$  
 formula $\alpha_g(x)$ in one free variable,  such that $\mathcal{A},g' \vDash \alpha_g(x)$ if and only if $g=g'$.
 \end{defi}  

 For any definable domain element $g$, we use $g$ as a constant symbol representing the domain element.

 \begin{obs}
 \label{obs:totalOrderDefble}
  For any definable family $\mathcal{F}_0$ of ($\mathcal{G},\leq_i,P_3$) which forms a total order under $\leq_i$, every member of $\mathcal{F}_0$ is definable as a constant.
  \end{obs}
 To see this, first observe that there exists a minimum element $f_1$ in $\mathcal{F}$ by well foundedness of the order $\leq_i$.
  \[ f_1(x) := \mathcal{F}_0(x) \; \wedge \; (  \forall y \; \mathcal{F}_0(y) \supset (x \leq_i y))\]
  Assuming $f_n$ (the $n^{th}$ smallest element of $\mathcal{F}_0$) has been defined, $f_{n+1}$ can be defined as the unique cover of $f_n$ in $\mathcal{F}_0$. We will also use a stronger version of this observation in the case where $\leq$ is any definable, discrete well-ordering of a definable family $\mathcal{F}_0$.

 \begin{defi}[Definability of Predicates] 
 \label{def:definabilityOfPreds}
  An $n$-ary predicate $R$ is definable in a $\tau$-structure $\mathcal{A}$ if there is a formula $\phi(\bar{x})$ with $n$ free variables such that for any $n$-tuple $\bar{a} \in \mathcal{A}$, 
  \[R(\bar{a}) \iff \mathcal{A},\bar{a} \models \phi(\bar{x})\]
  We say a predicate is \emph{definable in arithmetic} if and only if it is definable in ($\mathbb{N},\phi_{+},\phi_{\times}$). 
 \end{defi}
 We will also talk about a predicate ``definable in graphs'' to mean definable in a particular structure over graphs, which will be clear from the context. Let $R(x_1,...,x_n)$ be an $n-$ary relation such that whenever $R(\bar{g})$ holds of a tuple $\bar{g}=(g_1,g_2,...,g_n)$, it is the case that $g_i \leq g_j$ in the graph order $\leq$ under consideration. We adopt the convention that variable $x_i$ always occurs before $x_j$ in the tuple $\bar{x}$. Consider the relation $comp(x,y)$ which holds if and only if $y$ is a component of $x$. Since $comp(x,y)$ implies $y \leq_i x$, the convention adopted implies that we will write $comp(x,y)$ and not $comp(y,x)$. 

 We use the symbol $\phi$ for arithmetical formulae and $\psi$ for graph formulae. In the case of a definable family of graphs $\mathcal{F}_0$, we will write $x_1,x_2 \in \mathcal{F}_0$ instead of $\mathcal{F}_0(x_1) \wedge \mathcal{F}_0(x_2)$ to simplify notation. References to numbers in graph formulae are to be understood as the appropriate member of the family $\mathcal{N}$ i.e. the number $k$ is the graph $N_k$.  

We adopt the following convention to help the reader in parsing the formulae in this paper.  
\begin{rem}
  \label{rem:variableOrderConvention}
  We describe a convention we follow throughout this paper regarding the order of variables used to define predicates over graphs. Let $R(x_1,x_2,...,x_k)$ be a $k-$ary relation defined using a formula $\psi_R(x_1,x_2,...,x_k)$ over graph order. The order of the variables is $x_1 < x_2 < ... < x_k$ in this instance. The convention adopted states that if $R(g_1,g_2,...,g_k)$ implies that $g_i \leq g_j$ for some $i,j \in \mathbb{N}$, then it must be the case that $j < i$. For a more concrete example, consider the predicate $comp(x,y)$ which holds if and only if $y$ is a component of $x$. If $comp(g_1,g_2)$ holds, then so does $g_1 \leq_i g_2$. We could have chosen to define the predicate as: $comp(x,y)$ holds if and only $x$ is a component of $y$, but the former definition is used in accordance with the convention. 
\end{rem}

We make a note of the notation used in this paper in the following remark.
\begin{rem}
\label{rem:notation}
All variables $x,y,z$ occuring in formulae denote graphs and not labelled graphs. We will however need to talk about specific vertices or edges inside a graph and thus will require a labelling. So we will abuse notation and use $u_i,u_j$ to talk of vertices of a graph (not a labelled one) and  $u_iu_j$ for the edge joining $u_i$ and $u_j$. Usually the labelling on vertices will be clear from the context. We will use capital letters $U,V$ to denote sets of vertices. Given $V_0 \subseteq V(g)$ of vertices of a graph $g$, we will write $g[V_0]$ to denote the graph induced by the set $V_0$ of vertices. 

 We will use $e$ to denote the edge of a graph. We denote graphs by $g,h$; and graph families by caligraphic letters such as $\mathcal{P},\mathcal{C}$.

 We will denote by $N_i,K_i,C_i,S_i,P_i$ the graph consisting of $i$ isolated vertices, the $i$-clique, the cycle on $i$ vertices, the star on $i$ vertices and the path on $i$ vertices respectively (see Figure \ref{fig:NKPCS5}); and by $\mathcal{N},\mathcal{K},\mathcal{C},\mathcal{S},\mathcal{P}$ the corresponding families of isolated vertices, cliques, cycles, stars and paths. We denote the cardinality (number of vertices) of a graph $g$ by $|g|$, the number of edges of $g$ by $||g||$ and the disjoint union of graphs $g$ and $h$ by $g \cup h$.
\end{rem}

 \begin{figure}
 \begin{tikzpicture}
 \def \r {0.8cm}
 \def \s {0.5}
 \def \d {2.75}
 \def \et {1}
 \tikzstyle{node1} = [draw, circle, fill=c01, scale=\s]
 \tikzstyle{line1}=[draw,line width=\et,color=black!60]
  \begin{scope}
  \node[label=below:\Large $N_5$]{
  \begin{tikzpicture}
	\node[node1] (v1) at ({72*0}:\r) {};
	\node[node1] (v2) at ({72*1}:\r) {};
	\node[node1] (v3) at ({72*2}:\r) {};
	\node[node1] (v4) at ({72*3}:\r) {};
	\node[node1] (v5) at ({72*4}:\r) {};
  \end{tikzpicture}
  };
  \end{scope}
 
 \begin{scope}[shift={(\d,0)}]
 \node[label=below:\Large $P_5$]{
  \begin{tikzpicture}
	\node[node1] (v1) at ({72*0}:\r) {};
	\node[node1] (v2) at ({72*1}:\r) {};
	\node[node1] (v3) at ({72*2}:\r) {};
	\node[node1] (v4) at ({72*3}:\r) {};
	\node[node1] (v5) at ({72*4}:\r) {};
	\draw[line1] (v1) -- (v2) --(v3)--(v4)--(v5);
  \end{tikzpicture}
  };
 \end{scope}
 \begin{scope}[shift={(2*\d,0)}]
 \node[label=below:\Large $C_5$]{
  \begin{tikzpicture}
	\node[node1] (v1) at ({72*0}:\r) {};
	\node[node1] (v2) at ({72*1}:\r) {};
	\node[node1] (v3) at ({72*2}:\r) {};
	\node[node1] (v4) at ({72*3}:\r) {};
	\node[node1] (v5) at ({72*4}:\r) {};
	\draw[line1] (v1) -- (v2) --(v3)--(v4)--(v5)--(v1);
  \end{tikzpicture}
  };
 \end{scope}
 \begin{scope}[shift={(3*\d,0)}]
 \node[label=below:\Large $K_5$]{
  \begin{tikzpicture}
	\node[node1] (v1) at ({72*0}:\r) {};
	\node[node1] (v2) at ({72*1}:\r) {};
	\node[node1] (v3) at ({72*2}:\r) {};
	\node[node1] (v4) at ({72*3}:\r) {};
	\node[node1] (v5) at ({72*4}:\r) {};
	\draw[line1] (v1) -- (v2) --(v3)--(v4)--(v5)--(v1)--(v3)--(v5)--(v2)--(v4)--(v1);
  \end{tikzpicture}
  };
 \end{scope}
 \begin{scope}[shift={(4*\d,0)}]
 \node[label=below:\Large $S_5$]{
  \begin{tikzpicture}
	\node[node1] (v1) at ({72*0}:\r) {};
	\node[node1] (v2) at ({72*1}:\r) {};
	\node[node1] (v3) at ({72*2}:\r) {};
	\node[node1] (v4) at ({72*3}:\r) {};
	\node[node1] (v5) at ({72*4}:\r) {};
	\draw[line1] (v1) -- (v2)--(v3);
	\draw[line1] (v4)--(v2)--(v5);
  \end{tikzpicture}
  };
 \end{scope}

 \end{tikzpicture}
 \caption{Isolated points, path, cycle, clique and star of order 5 from left to right.} 
 \label{fig:NKPCS5}
\end{figure}
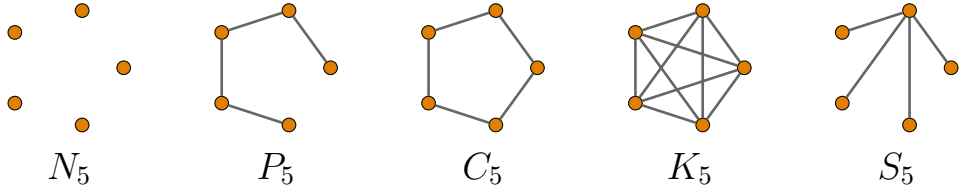
 \subsection{Recursive Predicates over Graphs}
 Next we have the definitions we need to formalize the meaning of ``recursive predicate over graphs.''
 There exist notions of computability and recursive predicates over abstract structures (see \cite{fitting2011fundamentals}), but for our purposes, we use a fixed encoding of graphs as strings so that the standard notion of a computable predicate as one accepted by a Turing machine can be used. We encode graphs as numbers (equivalently binary strings). These encodings were originally introduced in \cite{ramanujam2016definability}.

 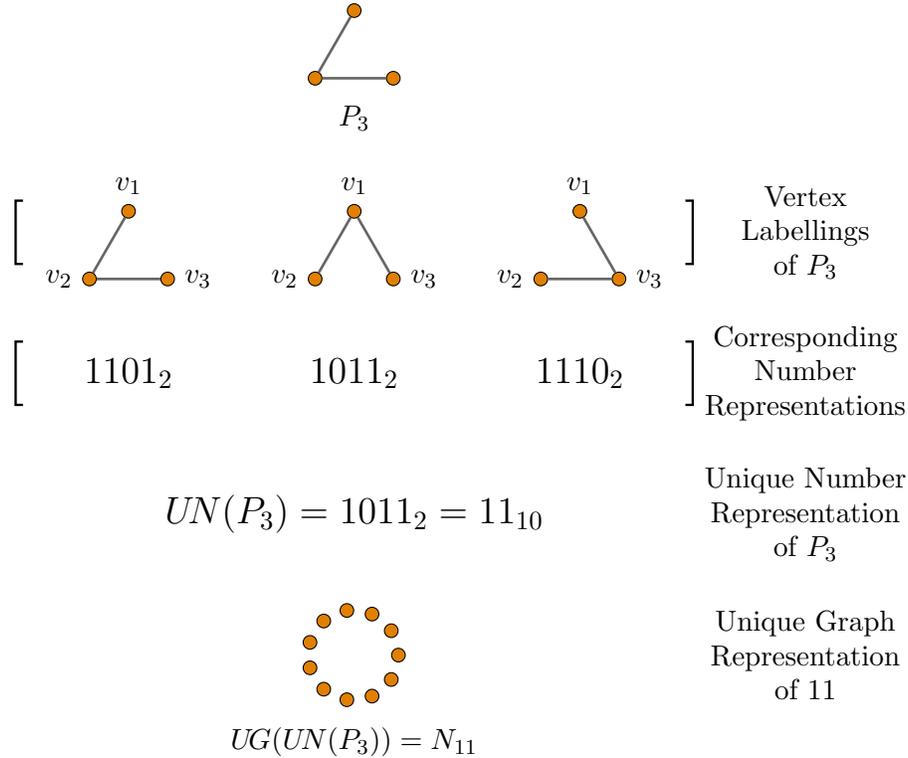
\begin{figure}
 \begin{tikzpicture}
 \def \r {0.6cm}
 \def \s {0.5}
 \def \hd {3}
 \def \vd {2.5}
 \def \et {1}
 \tikzstyle{node1} = [draw, circle, fill=c01, scale=\s]
 \tikzstyle{line1}=[draw,line width=\et,color=black!60]

   \begin{scope}
  \node[label=below:$P_3$ ]{
  \begin{tikzpicture}
  \node[node1] (v1) at (90:\r) {};
  \node[node1] (v2) at (210:\r) {};
  \node[node1] (v3) at (330:\r) {};
  \draw[line1] (v1) -- (v2) --(v3);
  \end{tikzpicture}
  };
  \end{scope}

  \node at (-\hd-1.5,-\vd) {\Huge $[$};
  \begin{scope}[shift={(-\hd,-\vd)}]
    \node{
  \begin{tikzpicture}
 \node[node1,label=above:$v_1$] (v1) at (90:\r) {};
  \node[node1,label=left:$v_2$] (v2) at (210:\r) {};
  \node[node1,label=right:$v_3$] (v3) at (330:\r) {};
  \draw[line1] (v1) -- (v2) --(v3);
  \end{tikzpicture}
  };
  \end{scope}
  \begin{scope}[shift={(0,-\vd)}]
  \node{
  \begin{tikzpicture}
  \node[node1,label=above:$v_1$] (v1) at (90:\r) {};
  \node[node1,label=left:$v_2$] (v2) at (210:\r) {};
  \node[node1,label=right:$v_3$] (v3) at (330:\r) {};
  \draw[line1] (v2) -- (v1) --(v3);
  \end{tikzpicture}
  };
  \end{scope}
  \begin{scope}[shift={(\hd,-\vd)}]
  \node[label ]{
  \begin{tikzpicture}
  \node[node1,label=above:$v_1$] (v1) at (90:\r) {};
  \node[node1,label=left:$v_2$] (v2) at (210:\r) {};
  \node[node1,label=right:$v_3$] (v3) at (330:\r) {};
  \draw[line1] (v2) -- (v3) --(v1);
  \end{tikzpicture}
  };
  \end{scope}
  
  \node[align=center] at (\hd+3,-\vd) {Vertex \\ Labellings \\ of $P_3$};

  \node at (\hd+1.5,-\vd) {\Huge $]$};
  \begin{scope}[shift={(0,-1.75*\vd)}]
    \node at (-\hd-1.5,0) {\Huge [};
    \node at (-\hd,0) {\Large $1101_2$};
    \node at (0,0) {\Large $1011_2$};
    \node at (\hd,0) {\Large $1110_2$};
    \node at (\hd+1.5,0) {\Huge ]};
  \end{scope}
  \node[align=center] at (\hd+3,-1.75*\vd) {Corresponding \\ Number \\ Representations};
  \node at (0,-2.5*\vd) {\Large $U\!N(P_3)=1011_2=11_{10}$};
  \node[align=center] at (\hd+3,-2.5*\vd) {Unique Number \\ Representation \\ of $P_3$};
  \begin{scope}[shift={(0,-3.25*\vd)}]
  \node[label=below:{$U\!G(U\!N(P_3))=N_{11}$} ]{
  \begin{tikzpicture}
     \foreach \i in {1,...,11}
     {
     \node[node1]  at ({360/11 * (\i  - 1)}:\r) {};
     }
  \end{tikzpicture}
  };
  \end{scope}
  \node[align=center] at (\hd+3,-3.25*\vd) {Unique Graph \\ Representation \\ of 11};

 \end{tikzpicture}
 \caption{From top to bottom we see how to obtain the graph $U\!G(U\!N(g)) \in \mathcal{N}$ from any graph $g \in \mathcal{G}$. The subscript in the numbers correpond to the base. }
 \label{fig:noRepOfGraphs}
\end{figure}

\begin{defi}[Number Representation of a Graph]
  The empty graph $\emptyset_g$ is represented by the number $0$ and the graph $N_1$ is represented by the number $1$.
  A number representation of a graph $g$ which is not $\emptyset_g$ or $N_1$ is defined using the following procedure. 
  \begin{enumerate}
    \item Choose a labelled graph $g'$ such that $g=[g']$. The order on vertices given by $L_{g'}$ induces an order $\leq_{lex}$ on set $S$ of all tuples of vertices ($v_i,v_j$) of $g$ with $j <i$. Let ($v_i,v_j$) and ($v_k,v_l$) belong to $S$ (i.e. $j <i, l < k$). Then ($v_i,v_j$) $\leq_{lex}$ ($v_k,v_l$) if and only if $i < k$ or $i =k, j < l$.
    \item Arrange all the tuples belonging to $S$ in descending order by $\leq_{lex}$ to form the sequence $seq$. 
  \item Create the number $m$ whose binary expansion is ${n \choose 2} +1$ bits long and has the following property:
   the $i^{th}$ most significant bit is 0 or 1 according to whether the $i-1^{th}$ tuple in $seq$ corresponds to a non-edge or edge (respectively) of the labelled graph $g'$. 
  \end{enumerate}
  The number $m$ is called a number representation of the graph $g$.
 \end{defi}
 We note that the number representation of a graph always has $1$ as its leading digit with the exception of $0$ which is used to represent $\emptyset_g$. 

 \begin{defi}[Unique Number Representation of a Graph aka $U\!N$]
 \label{defi:uniqueNumber}
  The unique number representation of a graph $g$ is the least number $m$ such that it is a number representation of $g$ and is denoted $U\!N(g)$. Note that the map $U\!N:~\mathcal{G}~\rightarrow~\mathbb{N}$ is a one-one map.(See Figure \ref{fig:noRepOfGraphs} for an example.)
  We extend the definition to tuples of graphs $\bar g$ in the usual way.
 \end{defi}

 \begin{obs}
 \label{obs:UNorder}
  The representation $U\!N$ induces an ordering on the vertices of the graph which comes from the underlying labelled graph.
 \end{obs}
 We can finally state what we mean by recursive predicates over graphs.
 \begin{defi}
 \label{defi:refPredGraphs}
  We say a predicate $R \subseteq \mathcal{G}^n$ is recursive if there exists a halting Turing machine $M$ such that
  \[ R(\bar{g}) \iff U\!N(\bar{g}) \in L(M)\] 
  i.e. the halting Turing machine $M$ accepts exactly the tuples of strings which correspond to $U\!N$ encodings of tuples belonging to $R$.
 \end{defi}

 In order to prove our main theorem, we will also need to encode numbers as graphs.

 \begin{defi}[Unique Graph Representation of a Number aka $U\!G$]
  \label{defi:uniqueGraph}
  Let $\mathcal{N}$ be the family of graphs which consists of graphs with no edges. This family is totally ordered by $\leq_s$ as well as $\leq_i$ and contains exactly one graph of cardinality $k$, denoted by $N_k$.

  The one-one map $U\!G : \mathbb{N} \rightarrow \mathcal{G}$ sends a number $k$ to the graph $N_k$ which is called the unique graph representation of $k$.
 \end{defi}

\subsection{Arithmetical Predicates over Graphs}



  We define the notion of an \textit{arithmetical} structure over graphs. For the sake of convenience, we will only deal with relational structures. 

\begin{defi}
Fix a total order $\leq_t$ on $\mathcal{G}$ such that ($\mathcal{G},\leq_t$) is isomorphic to ($\mathbb{N},\leq$). Let $plus_t$ and $times_t$ be the ternary addition and multiplication relations with respect to the order $\leq_t$. A relation $R$ over graphs is called an \textit{arithmetical} graph relation if it is definable in first order logic over ($\mathcal{G},plus_t,times_t$). 

A relational structure $(\mathcal{G},\tau)$ over graphs is called \textit{arithmetical} if every relation $R \in \tau$ is arithmetical. 
\end{defi}
We will drop `graph' from `arithmetical graph relation' and just say `arithmetical relation' when it is clear from context. 

\begin{rems}
\label{rem:numberTheoreticPred}
\begin{enumerate}
  \item Graph relations which are constrained to be subsets of $\mathcal{N}^k$ for some value of $k$ will be called \textit{number theoretic} predicates in contrast with the definition of arithmetical predicates above.
  \item Arithmetical relations over graphs are analogous to relations over numbers belonging to the arithmetical hierarchy and thus include recursive relations over graphs.
\end{enumerate}
\end{rems}

\begin{obs}
\label{obs:arithStrArithPred}
  Every predicate definable in an arithmetical structure is also arithmetical.
\end{obs}


\begin{defi}
  An arithmetical structure ($\mathcal{G},\tau$) over graphs is said to have the \textit{maximal definability property} if every arithmetical predicate (w.r.t. $\leq_t$) is definable in the first order theory of  ($\mathcal{G},\tau$).
\end{defi}
The maximal definability property for an arithmetical structure ($\mathcal{G},\tau$) gives a characterization of the definable predicates of the structure i.e. a predicate is definable in ($\mathcal{G},\tau$) if and only if it is arithmetical (with respect to $\leq_t)$. The set of arithmetical predicates with respect to any other total order $\leq_{t'}$ on $\mathcal{G}$ remains the same as long as $\leq_{t'}$ is arithmetical with respect to $\leq_t$. The situation is analogous to the considerations of string encodings of graphs encountered in the previous section; the set of recursive predicates over graphs with respect to two different encodings of graphs as strings is the same as long as there exists a halting Turing Machine which can translate one string encoding to the other. 

Next we define the total order $\leq_t$ with respect to which the structures we consider in this manuscript are arithmetical. 

\begin{defi}
\label{defi:totalOrderOnGraphs}
Define $g_1 \leq_t g_2$ for $g_1,g_2 \in \mathcal{G}$ if and only if $U\!N(g_1) \leq U\!N(g_2)$. Note that $\leq_t$ is a total order on $\mathcal{G}$ by the uniqueness of the map $U\!N$.
\end{defi}

\section{Definability of Arithmetical Predicates in an Arbitrary Structure over Graphs.}
\label{sec:condResultGenStr}
 We state and prove sufficient conditions for an arithmetical structure ($\mathcal{G},\tau$) over graphs to have the maximal definability property. As a corollary we also get the fact that any capable structure can define every recursive predicate over graphs. These conditions are: definability of arithmetic, definability of cardinality of a graph and definability of two relations concerning vertex labelled representations of graphs called \emph{o-presentations}, in ($\mathcal{G},\tau$). These o-presentations were first introduced by Je\v{z}ek and Mckenzie \cite{jevzek2010definability}, and defined for graphs by Wires \cite{wires2016definability}.

\begin{figure}[ht]
 \begin{tikzpicture}
 \def \r {0.7cm}
 \def \s {0.5}
 \def \d {2cm}
 \def \et {1}
 \def \el {0.4}
 \tikzstyle{node1} = [draw, circle, fill=c01, scale=\s]
 \tikzstyle{line1}=[draw,line width=\et,color=black!60]
 \node[shift={(0,\d)},label=below:$S_4$]{
  \begin{tikzpicture}
   \node[node1] (v1) at (0,0) {};
   \node[node1] (v2) at ($(v1)+(0,1.5*\el)$) {};
   \node[node1] (v3) at ($(v1)+(210:1.5*\el)$) {};
   \node[node1] (v4) at ($(v1)+(330:1.5*\el)$) {};
   \draw[line1] (v2)--(v1)--(v3);
   \draw[line1] (v1)--(v4);
  \end{tikzpicture}
 };
   \node[shift={(0,-\d)},label=below: Labelled $S_4$]{
    \begin{tikzpicture}
   \node[node1,label=below:$v_1$] (v1) at (0,0) {};
   \node[node1,label=above:$v_2$] (v2) at ($(v1)+(0,1.5*\el)$) {};
   \node[node1,label=left:$v_3$] (v3) at ($(v1)+(210:1.5*\el)$) {};
   \node[node1,,label=right:$v_4$] (v4) at ($(v1)+(330:1.5*\el)$) {};
   \draw[line1] (v2)--(v1)--(v3);
   \draw[line1] (v1)--(v4);
  \end{tikzpicture}
    };
  \node[shift={(3*\d,0)},label=below:an o-presentation of $S_4$]{
   \begin{tikzpicture}
    \node[node1] (v1) at (0,0) {};
    \node[node1] (v2) at ($(v1)+(0,1.5*\el)$) {};
    \node[node1] (v3) at ($(v1)+(210:1.5*\el)$) {};
    \node[node1] (v4) at ($(v1)+(330:1.5*\el)$) {};
    \draw[line1] (v2)--(v1)--(v3);
    \draw[line1] (v1)--(v4);

   \begin{scope}[shift={(45:\d)}] 
    \draw[line1] (0,0) circle [radius=\r];
    \foreach \i in {1,...,8}
    {
    \node[node1] (\i) at ({360/8 * (\i  - 1)}:\r) {};
    }
   \end{scope}
   \draw[line1] (v2)--(6);
   \node at (1.5,2.5) {$C_8$};
   \begin{scope}[shift={(135:\d)}] 
    \draw[line1] (0,0) circle [radius=\r];
    \foreach \i in {1,...,7}
    {
    \node[node1] (\i) at ({360/7 * (\i  - 1)}:\r) {};
    }
   \end{scope}
   \draw[line1] (v1)--(7);
   \node at (-1.5,2.5) {$C_7$};
   \begin{scope}[shift={(225:\d)}] 
    \draw[line1] (0,0) circle [radius=\r];
    \foreach \i in {1,...,9}
    {
    \node[node1] (\i) at ({360/9 * (\i  - 1)}:\r) {};
    }
   \end{scope}
   \draw[line1] (v3)--(2);
   \node at (-1.5,-2.5) {$C_9$};
   \begin{scope}[shift={(315:\d)}] 
    \draw[line1] (0,0) circle [radius=\r];
    \foreach \i in {1,...,10}
    {
    \node[node1] (\i) at ({360/10 * (\i  - 1)}:\r) {};
    }
   \end{scope}
   \draw[line1] (v4)--(5);
   \node at (1.5,-2.5) {$C_{10}$};
   \end{tikzpicture}
 };

 \end{tikzpicture}
  \caption{Top left: the graph $S_4$. Bottom left: A vertex labelling of $S_4$. Right: The o-presentation corresponding to the given vertex labelling.}
 \label{fig:opresentation}
\end{figure}
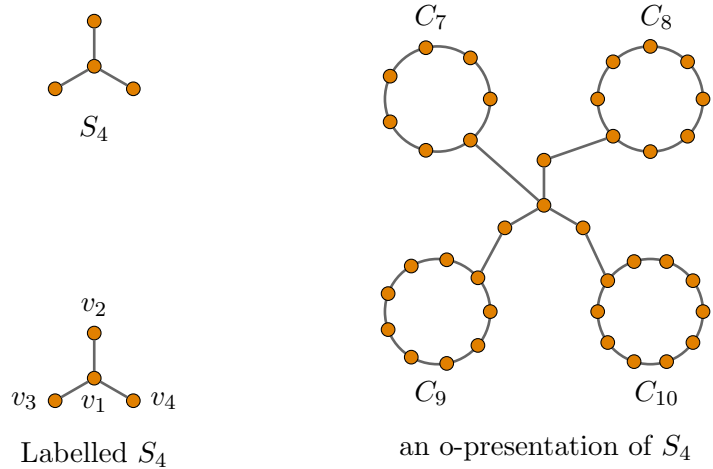

 \begin{defi}[o-presentation]
 \label{def:opres}
	An o-presentation of $g \in \mathcal{G}$ is another graph $g'$ constructed as follows:
	Fix a vertex labelling $v_1,v_1,..,v_n$ of vertices of $g$. Let $g''$ be the graph formed by the disjoint union of $g$ and the cycles $C_{n+i+2}$ for each $1 \leq i \leq n$. Add $n$ additional edges to $g''$ connecting each cycle $C_{n+i+2}$ to the corresponding vertex $v_i$. The resulting graph is $g'$ (see Figure \ref{fig:opresentation}).
 \end{defi}
 \begin{rem}
	Note that each vertex labelling of a graph $g$ leads to a (possibly) different o-presentation. We will refer to the set of o-presentations of $g$ by $\tilde{g}$ and write $g' \in \tilde{g}$ to indicate that $g'$ is an o-presentation of $g$. The example in Figure \ref{fig:opresentation} clarifies the bijective correspondence between o-presentations and labellings of a graph.
 \end{rem}
 \begin{defi}[Indicator Cycle]
	\label{def:indicatorCycle}
	Given a graph $g$, a cycle $C$ is called an \textit{indicator cycle} of $g$ if $|g|+3 \leq |C| \leq 2|g|+2$. 
 \end{defi}
 We will often just call $C$ an indicator cycle if the graph $g$ is understood from the context.

 \begin{defi}[Capable Structure over Graphs]
	\label{def:capableStructure}
	We call an arithmetical structure ($\mathcal{G},\tau$) a \textit{capable} structure over graphs if it satisfies the following three conditions:
    \begin{description}
      \item[($C1$)] Arithmetic can be defined in ($\mathcal{G},\tau$); in particular, the following predicates are definable:
	 \begin{enumerate}
		\item The family $\mathcal{N}$ of graphs which do not contain edges i.e. are made of isolated points.
		\item The predicate $\psi_{+}(x,y,z)$ which holds if and only if $x,y,z \in \mathcal{N}$ and $|x|+|y|=|z|$.
		\item The predicate $\psi_{\times}(x,y,z)$ which holds if and only if $x,y,z \in \mathcal{N}$ and $|x| \times |y| =|z|$.
	 \end{enumerate}
	\item[($C2$)] The following predicates related to o-presentations are definable in ($\mathcal{G},\tau$):
	 \begin{enumerate}
	 \item The predicate $\psi_{opres}(y,x)$ which holds if and only if $y$ is an o-presentation of $x$, also written $y \in \tilde{x}$.
	 \item The predicate $\psi_{edgeOP}(x,i,j)$ which holds if and only if $i,j \in \mathcal{N}$ and there exists a graph $y$ such that $x \in \tilde{y}$ and in the vertex labelling induced on $y$ by the o-presentation $x$, the vertices $v_{|i|}$ and $v_{|j|}$ in $y$ have an edge.
	 \end{enumerate}
   \item[($C3$)] The predicate $sameCard(x,y)$ which holds if and only if $x$ and $y$ have the same number of vertices, is definable in ($\mathcal{G},\tau$).
   \end{description}
 \end{defi}
 \begin{obs}
 	\label{lem:cardN}
 	We will write $|x|_g$ to denote the member of $\mathcal{N}$ which has the same cardinality as $x$. Note that the binary relation $card(x,y)$ mapping a graph $x$ to $|x|_g=y$ is definable :
 	\[ card(x,y):= \mathcal{N}(y) \wedge sameCard(x,y)\]
 	Hence we will freely use $|x|_g$ in formulae over a capable structure. Where it is clear from context that we are referring to $|x|_g$ and not $|x|$, we will abuse notation by dropping the subscript. 
 \end{obs}

 \begin{rems}
 \begin{enumerate}
	\item The use of $\mathcal{N}$ as the representation of numbers in the above definition is only because it is a natural choice which aids understanding. The three graph orders $\leq_i,\leq_s,\leq_m$ considered in this paper are all equivalent over $\mathcal{N}$ and correspond to the order over the natural numbers.  Replacement of $\mathcal{N}$ by an arbitrary $\mathcal{D} \subseteq \mathcal{G}$ (with appropriate modifications to the requirements for defining addition and multiplication) does not affect the results in this section.
	\item Similarly, the definition of o-presentation used corresponds to conversion of edge information in a graph $g$ into subgraph information about its o-presentation $g'$, which is a natural way of extracting the internal structure in graph orders. The o-presentation appropriate for arbitrary structures over graphs depends on the structure in question; all that is required is the labelling of the vertices of a graph in a definable way and capture of edge information.
 \end{enumerate}	
 \end{rems}

 \begin{thm}
   \label{thm:capableImpliesMdp}
   Any structure ($\mathcal{G},\tau$) over graphs which is capable has the maximal definability property.
 \end{thm}

 To prove the theorem, we need to show that the predicates $plus_t$ and $times_t$ are definable using the relations $\tau$ for any arithmetical structure ($\mathcal{G},\tau$) over graphs which is capable.  Equivalently, we can prove that for every recursive predicate $R \subseteq \mathcal{G}^n$ over graphs, there exists a formula $\psi_{R}(\bar{x})$ (where $|\bar{x}|=n$)  over the vocabulary $\tau$ such that for any $n$-tuple of graphs $\bar{g}$, 
 \[ R(\bar{g}) \iff (\mathcal{G},\tau) \models \psi_{R}(\bar{g})\]
Since $plus_t$ and $times_t$ are recursive predicates, the result follows.
 \definecolor{c00}{RGB}{50,50,50} 
 \definecolor{c01}{RGB}{225,128,0} 
 \definecolor{c02}{RGB}{42, 108, 255} 
 \begin{figure}
  \begin{tikzpicture}

  \tikzstyle{circleG}=[draw, line width=1.5,color=c01]
  \tikzstyle{circleN}=[line width=1.5, color=c02]
  \tikzstyle{node1} = [draw, circle, fill, scale=0.3]
  \tikzstyle{fn} = [draw, -{Latex[length=10,width=8]}, thick, color=c00]

   \draw[circleG] (0,0) circle [radius=3];
   \node (G) at (120:3.3) {$\mathcal{G}$};

   \node at (2,1.5) {$R$};
   \begin{scope}[rotate=-45]
    \draw[circleG] (1,0.5) ellipse (1 and 2);
   \end{scope}

	\node at (-2,1.6) {$\mathcal{N}$};
    \node[rotate=-15] at (-1.3,.7) {\tiny $U\!G(U\!N(\mathcal{G}))$};

   \begin{scope}[rotate=45]    
    \draw[circleG] (-1,0) ellipse (1.5 and 2.5);
    
    \draw[circleG] (-1.3,0) ellipse (1 and 2);

    \draw[circleG] (-1.1,0) ellipse (.5 and 1);
    \path[postaction={decorate, 
    decoration={
        text along path,
        raise = 1ex,
        text={|\tiny| UG(UN(R))},
        text align=left
            }
        }] (-1,0) ++(175:0.5 and 1) arc(200:30:0.5 and 0.6);

   \end{scope}
   
   \node[node1,label=below:$g$] (g) at (2,0.5) {};
   \node[node1,label=left:$g'$] (UGUNg) at (-0.7,-0.7) {};
   
   \node[node1,label=below:$g_1$] (g1) at (0.5,2) {};
   \node[node1,label=left:$g_1'$] (UGUNg1) at (0,-2.2) {};

  \begin{scope}[xshift=8cm]
    \draw[circleN] (0,0) circle [radius=3];
    \node (G) at (60:3.3) {$\mathbb{N}$};
    
    \draw[circleN] (-0.6,0.6) circle [radius=2];
      \node[rotate=-45] at (1.1,2) {$U\!N(\mathcal{G})$};

    \draw[circleN] (-1,0.2) circle [radius=1];
    \path[postaction={decorate,
    		decoration={text along path, 
    					raise = 1ex,
    					text={UN(R)},
    					text align=left}}]
			(-1,0.2) ++(80:1)  arc [start angle=80, end angle=0, radius=1];
			
	\node[node1,label=below:\tiny $U\!N(g)$] (UNg) at (-1.2,0.5) {};
	\node[node1,label=right:\tiny $U\!N(g_1)$] (UNg1) at (1,1.1) {};
	
  \end{scope}
  
	\draw[bend left=20,fn]  (g) to node [auto] {$U\!N$} (UNg);
	\draw[bend left=10,fn]  (UNg) to node [auto] {$U\!G$} (UGUNg);
	
	\draw[bend left=20,fn]  (g1) to node [auto] {$U\!N$} (UNg1);
	\draw[bend left=40,fn]  (UNg1) to node [auto] {$U\!G$} (UGUNg1);
   
 \end{tikzpicture}
 \caption{Maps $U\!N$ and $U\!G$ and how they act on $R \subseteq \mathcal{G}$. For any graph $g \in R$ it is the case that $U\!G(U\!N(g))=g' \in U\!G(U\!N(R))$. For any graph $g_1 \notin R$ correspondingly $U\!G(U\!N(g_1))=g_1' \in (U\!G(U\!N(\mathcal{G})) \setminus U\!G(U\!N(R)))$.}
 \label{fig:UGUN}
 \end{figure}
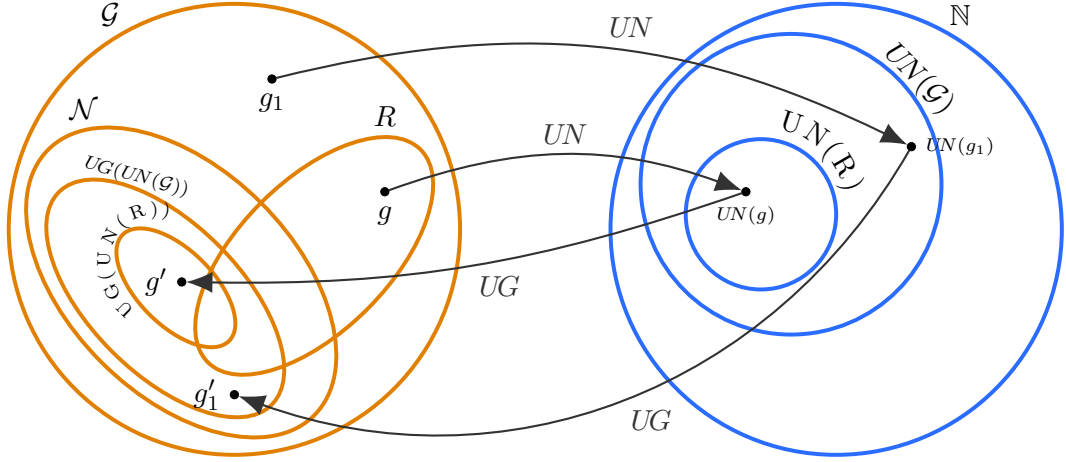

 We explain the proof idea behind the construction of $\psi_{R}$ assuming that $R$ is unary (see Figure \ref{fig:UGUN}). This consists of two major steps:

 \noindent \underline{Step 1 :}\\
  Corresponding to every graph $g$ there exists the graph $g'=U\!G(U\!N(g)) \in \mathcal{N}$. From the definitions of the maps $U\!N$ and $U\!G$ it is clear that the map $U\!G \circ U\!N$ is a bijection between the sets $\mathcal{G}$ and $U\!G(U\!N(\mathcal{G}))$. The fact that $U\!G(U\!N(\mathcal{G})) \subseteq \mathcal{N}$ allows us to use the definability of arithmetic to capture the image $U\!G(U\!N(R))$ of the relation $R$ using a formula $\psi^{trans}_{U\!G(U\!N(R))}$ over $\tau$. 

 \noindent \underline{Step 2 :}\\
  What remains to be done is to show that the predicate $\psi^{trans}_{enc}(x,y)$ which holds if and only if $y=U\!G(U\!N(x))$ is definable in ($\mathcal{G},\tau$). Once this is done, we immediately get the defining formula $\psi_{R}$:
	\[ \psi_{R}(x) := \exists y \; \psi_{enc}(x,y) \; \wedge \; \psi^{trans}_{U\!G(U\!N(R))}(y)\]
  The formula $\psi_{enc}$ uses condition $C2$ to verify that the edge information contained in $x$ and $y$ match.

  We now show how to define $\psi^{trans}_{U\!G(U\!N(R))}$.

  Since $R$ is a recursive predicate, by Definition \ref{defi:refPredGraphs} there exists a machine $M$ which accepts the $U\!N$ encodings of the set of graphs which belong to $R$. 
  \[ R(\bar{g}) \iff U\!N(\bar{g}) \in L(M)\]

 \noindent We recall the following result which is a weakening of the MRDP theorem \cite{matiyasevich1971diophantine}.
 \begin{thm}
 \label{thm:recInArith}
	Every recursive predicate $R$ on numbers is definable in arithmetic.
 \end{thm} 
 \noindent Thus there is an arithmetic formula $\phi_{U\!N(R)}(\bar{x})$ such that for any tuple $\bar{n}$ of numbers, 
 \[ (\mathbb{N}, \phi_{+}, \phi_{\times}) \models \phi_{U\!N(R)}(\bar{n}) \iff \bar{n} \in U\!N(R)\]

 This gives us an important corollary:
\begin{cor}
  The following relations over graphs are arithmetical:
  \begin{enumerate}
    \item Each of the graph orders $\leq_i,\leq_s,\leq_m$.
    \item The constant $P_3$.
    \item The relation $sameSize(x,y)$ which holds if and only if $x$ and $y$ have the same number of edges.
    \item The relation $disjointUnion(z,x,y)$ which holds if and only if $z$ is the disjoint union of $x$ and $y$.
  \end{enumerate}
  This implies that all structures over graphs dealt with in this paper are arithmetical.
\end{cor}

 The condition $C1$ that arithmetic is definable gives the following corollary:
 \begin{cor}
 \label{cor:translateArithInduced}
	For every formula $\phi(\bar{x})$ in arithmetic there is a formula $\psi^{trans}_{U\!G}(\bar{x})$ in ($\mathcal{G},\tau$) such that 
	\[(\mathbb{N},\phi_{+},\phi_{\times}) \models \phi(\bar{n}) \iff (\mathcal{G},\tau) \models \psi^{trans}_{U\!G}(U\!G(\bar{n})) \]
 \end{cor}
 The above corollary, which can be used to exploit the power of arithmetic now available to us in vocabularies over graphs, will be critically used later. 

 Applying this translation to $\phi_{U\!N(R)}$ gives us the graph formula $\psi^{trans}_{U\!G(U\!N(R))}$. Such translated formulae define relations which we call number theoretic (see Remark \ref{rem:numberTheoreticPred}).

 We now show how to define $\psi_{enc}$. To do so, we need number theoretic predicates to identify the image set $U\!N(\mathcal{G}) \subseteq \mathbb{N}$, identify the cardinality of $g$ and draw out the edge information $E(g)$ from $U\!N(g)$. 
 This has already been accomplished in previous work:
 \begin{thm}[\cite{ramanujam2016definability}]
 \label{thm:reqGraphPredInArith}
	The following predicates are definable in arithmetic:
	\begin{enumerate}
		\item  $\phi_{U\!N}(x)$ holds if and only if $x$ is a number which represents a graph as given in Definition \ref{defi:uniqueNumber}.
		\item $\phi_{edge}(x,i,j)$ holds if and only if there is a graph $g$ such that $x=U\!N(g)$ and $v_iv_j$ is an edge in the order induced by the map $U\!N$.
		\item $\phi_{length}(n,x)$ holds if and only if the length of the binary representation of $x$ is $n$. (We will just write $length(x)$ to denote $n$.)
	\end{enumerate}
 \end{thm}
 We can now define the binary relation $y=U\!G(U\!N(x))$ by the formula $\psi_{enc}(x,y)$:  
 \begin{flalign*}
	 	\psi_{enc}(x,y) := & \mathcal{N}(y) \; \wedge \; \psi^{trans}_{graphOrder}(y,|x|_g) \; \wedge \; \psi^{trans}_{U\!N}(y) \; \wedge \\
	     &\exists z \;[ z \in \tilde{x} \; \wedge \; \forall i,j \in \mathcal{N}\; ( N_1 \leq_i i <_i j \leq_i |x|_g) \supset \\ 
	 	& \quad \psi^{trans}_{edge}(y,i,j) \iff 
	 	 \psi_{edgeOP}(z,i,j) ]
 \end{flalign*} 
 The formula is composed of the following subformulae:
 \begin{enumerate}
   	\item $\psi^{trans}_{U\!N}$ and $\psi^{trans}_{edge}$ are translations of arithmetical formulae from Theorem \ref{thm:reqGraphPredInArith}.
   	\item $z \in \tilde{x}$ and $\psi_{edgeOP}(z,i,j)$ are from condition $C2$.
   	\item The graph $|x|_g \in \mathcal{N}$ which has the same cardinality as $x$ uses condition $C3$.
   	\item $\psi^{trans}_{graphOrder}$ is the translation of the following arithmetic formula:
 \[\phi_{graphOrder}(n,m) := length (n) = 1+m(m-1)/2\].
 \item The use of the induced subgraph order $\leq_i$ in $ N_1 \leq_i i <_i j \leq_i |x|_g$ follows from the fact that the order $\leq$ among numbers is definable using addition and ($\mathcal{N},\leq_i$) is isomorphic to ($\mathbb{N},\leq$).
   \end{enumerate}

	Given a graph $x$, $U\!N(x)=n_x$ is a number which has bit length $1+|x|(|x|-1)/2$. Applying $U\!G$ to $n_x$ should give us $y$. Thus the number $n=U\!G^{-1}(y)$ should have bit length $1+|x|(|x|-1)/2$. This condition is taken care of by the formula $\psi^{trans}_{graphOrder}(y,|x|)$. Effectively, what this amounts to is identifying the cardinality of $x$. The formula $\psi^{trans}_{U\!N}(y)$ verifies that $n$ belongs to the set $U\!N(\mathcal{G})$. Finally, we need to check that the edge information is correct to conclude that $n_x=n$. This is done using the fact that there is a witnessing o-presentation $z$ of $x$ such that the edge information in $z$ matches with the edge information in $n$ (which is accessed via $y$ in the formula). 
	This concludes the definability of $\psi_{enc}(x,y)$.
 
 We can now write the required formula $\psi_R$ in ($\mathcal{G},\tau$). In the general case of an $n-ary$ predicate $R$ it takes the form:
 \[ \psi_{R}(\bar{x}) := \exists \bar{y} \; \bigwedge_{i=1}^n \psi_{enc}(x_i,y_i) \; \wedge \; \psi^{trans}_{U\!G(U\!N(R))}(\bar{y})\]
 This concludes the proof of Theorem \ref{thm:capableImpliesMdp}.

\begin{cor}
\label{cor:capableMeansRecPred}
  For any ($\mathcal{G},\tau$) which is a capable structure over graphs, every recursive predicate $R \subseteq \mathcal{G}^n$ over graphs is definable in ($\mathcal{G},\tau$).
\end{cor}

\section{Defining Arithmetical Predicates in Graph Orders}
\label{sec:graphOrders}
 In this section, we will show that the structures ($\mathcal{G},\leq_i,P_3)$ , ($\mathcal{G},\leq_s$) and ($\mathcal{G},\leq_m,sameSize$) are capable. The proof uses basic predicates already known to be definable in the induced subgraph order \cite{wires2016definability} and the subgraph order \cite{ramanujam2016definability}. In the case of the subgraph order, we give defining formulae in Appendix \ref{app:basicPredSubgraph} and for proof of correctness refer the reader to \cite{ramanujam2016definability}. In the case of basic predicates definable in the induced subgraph, we refer the reader to \cite{wires2016definability}. The result for ($\mathcal{G},\leq_m,sameSize$) follows as a corollary of the result for the subgraph order. 

 Recall from Definition \ref{def:capableStructure} that a capable structure satisfies three conditions $C1$,$C2$ and $C3$. Condition $C1$ is known to hold for the induced subgraph and subgraph orders \cite{wires2016definability,ramanujam2016definability} and $C3$ follows from the observation below:

\begin{obs}
	\label{obs:C3fromC1}
	The predicate $sameCard(x,y)$ is definable in each of the structures ($\mathcal{G},\leq_i,P_3)$ , ($\mathcal{G},\leq_s$) and ($\mathcal{G},\leq_m$) assuming the definability of the family $\mathcal{N}$:
	\[ sameCard(x,y) := \forall z \; \mathcal{N}(z)  \supset z \leq x \iff z \leq y\]
	where $\leq$ is the appropriate order.
\end{obs}

 We will concentrate on showing that $C2$ holds.\\
 The proof strategy for showing that $C2$ holds involves constructing an o-presentation $y \in \tilde{x}$ from the disjoint union of $x$ with another graph $g$. However, this graph $g$ is a certain disjoint union of paths in the induced subgraph order, while it is a disjoint union of large cycles in the subgraph order. The difference arises from the differences in the covering relation alluded to in Remark \ref{rem:coverSubVsInd}. The differences in the predicates which can be easily defined in one order as compared to the other also lead to differences in the kind of intermediate predicates constructed. We deal with the induced subgraph order first and later get to the subgraph order.  

\subsection{Induced Subgraph Order} 
\label{sub:induced_subgraph_order}
The fact that the induced subgraph order satisfies the three conditions ($C1$),($C2$) and ($C3$) follows as a corollary of Wires' result (see Theorem 12.3 in Wires \cite{wires2016definability}) on the definability of every graph property in the first order theory of the associated small category of graphs. Many of the explicit formulae and intermediate relations required are also already present in his work. In this section, we flesh out the details, giving explicit formulae in case they are not available in literature or when they are useful in understanding the strategy employed. The strategy carries over to the subgraph order dealt with in the next section with some small changes. We do not give explicit formulae for some graph theoretical relations of independent interest defined by Wires and in this case refer the reader to Wires \cite{wires2016definability} for details. 

Beginning with this section, many number theoretic predicates will be used and we extend our abuse of notation from dropping the subscript in $|x|_g$ to the use of subformulae such as $n+1 < m < 2n+3$, $m=n^2+4$ and $C_n \leq_i x$; in all such cases, there is a subformula of the form $m,n \in \mathcal{N}$ stating that the variables are members of $\mathcal{N}$, and these variables are treated as numbers with the understanding that corresponding formulae can be written without the abuse of notation by condition ($C1$).   

\underline{($\mathcal{G},\leq_i,P_3$) satisfies $C1$}:\\
 This is implicit in the work of Wires \cite{wires2016definability}. Only the formula for the definability of multiplication remains to be shown.
 \begin{lem}[Wires \cite{wires2016definability}]
 \label{lem:basicPredInduced1}
  The following predicates are definable in the induced subgraph order:
  \begin{enumerate}
	\item The families $\mathcal{N},\mathcal{T}, \mathcal{P}$ of isolated points, trees and paths  respectively.
	\item $\psi_{+}(x,y,z)$ holds if and only if $x,y,z \in \mathcal{N}$ and $|x|+|y|=|z|$.
	\item $|x|=|y|$ holds if and only if $x$ and $y$ have the same cardinality (i.e. same number of vertices, also known as order of the graph).
  \end{enumerate}
 \end{lem}
 It remains to define the predicate $\psi_{\times}(x,y,z)$.
 Instead of the formula $\psi_{\times}(x,y,z)$, we can equivalently define the square predicate $\psi_{sq}(x,y)$ which holds if and only if $x,y \in \mathcal{N}$ and $|x|=|y|^2$. To do so, we construct a tree $t_n$ given a numerical parameter $n$. The tree $t_n$ has $V(t_n)=\{v_0\} \cup \{ v_1,v_2,...,v_n\} \cup \bigcup_{i=1}^n \{v_{i,1},v_{i,2},...,v_{i,n-1} \}$, $E(t_n)= \{ v_0v_1,v_0v_2,...,v_0v_n\} \cup \bigcup_{i=1}^n \{ v_iv_{i,1},v_iv_{i,2},...,v_iv_{i,n-1}\}$. Informally, $t_n$ has a root $v_0$ which has degree $n$ and each of its neighbours also has degree $n$. It is easy to see that $|t_n|=1+n+n(n-1)=n^2+1$.
 
 We need to define the family of stars, which are trees containing a vertex which is incident on all edges. See Figure \ref{fig:NKPCS5} for the star $S_5$ on five vertices. 
 \begin{obs}
 \label{obs:starInIndSub}
 	The predicate $\mathcal{S}(x)$ which holds if and only if $x$ is a star is definable in the induced subgraph order.
 	\[ \mathcal{S}(x) :=\mathcal{T}(x) \; \wedge \; P_4 \nleq_i x\]
  As usual, it is easy to see that the conditions specified are necessary. In a tree, a path is present as a subgraph iff it is present as an induced subgraph. Any graph containing $P_4$ as a subgraph cannot be a star.\\
  Next we define a tree which has maximum degree $n$ and maximum path subgraph $P_5$, and then take the largest (by order) such  tree. For any fixed $n$, there exists a unique such tree $t_n$. 
 \end{obs}
 
 \begin{flalign*}
	stree'(x,n) := &\mathcal{T}(x) \; \wedge \; maxDeg(x,n) \; \wedge \; maxPath5(x)\\
	&where\\
	&maxDeg(x,n) := \mathcal{N}(n) \; \wedge \; S_{n} \leq_i x \; \wedge \; \forall m\; (\mathcal{N}(m) \; \wedge \; n < m \supset S_{m} \nleq_i x)\\
	&maxPath5(x) :=  P_5 \leq_i x \; \wedge \; \forall m \; (\mathcal{N}(m) \; \wedge \; 5 < m \supset P_{m} \nleq_i x)\\
	stree(x,n) := &stree'(x,n) \; \wedge \; \forall y\; stree(y,n) \supset |x| \leq |y|
 \end{flalign*}

 Given $n$, the unique $x$ satisfying $stree(n,x)$ is the tree $t_n$ described above. 
 Let $x$ be a tree satisfying $stree'(n,x)$. Fix some arbitrary vertex $v_0$ as root of this tree. The degree condition implies that $v_0$ has at most $n$ neighbouring vertices and the $maxPath5$ condition ensures that the maximum depth of the tree is two. The tree of maximum cardinality satisfying $stree'$ is thus exactly the tree $t_n$.
  We can define the relation $\psi_{sq}(x,y)$ now:
 \[ \psi_{sq}(x,y) := \mathcal{N}(x) \; \wedge \; \mathcal{N}(y) \; \wedge \; \; \exists z \; stree(z,y) \; \wedge \; x \lessdot_i |z| \]
 This ends the proof that ($\mathcal{G},\leq_i,P_3$) satisfies condition $C1$.

 \underline{($\mathcal{G},\leq_i,P_3$) satisfies $C2$}:\\
 We assume that the predicate $\tilde{\mathcal{G}}(x)$ if and only if $x$ is an o-presentation is definable. The proof of definability of $\tilde{\mathcal{G}}$ is postponed to the end of this subsection. Note that one can directly define $\psi_{opres}$ and then define $\tilde{\mathcal{G}}$ in terms of $\psi_{opres}$; the route we take simplifies the presentation.

 We recall some basic predicates known to be definable in the induced subgraph order.  
 \begin{lem}[Wires \cite{wires2016definability}]
 \label{lem:basicPredInduced2}
	The following predicates are definable in ($\mathcal{G},\leq_i,P_3$).
	\begin{enumerate}
		\item The family $\mathcal{C}$ of cycles. 
		\item $maximalComp(y,x)$ holds if and only if $x$ is a maximal connected component of $y$.
		\item $cover(y,x,n)$ holds if and only if there are exactly $n-1$ graphs between $x$ and $y$ in the order and $x \leq_i y$. Also denoted $x \lessdot_i^n y$.
		\item $\mathcal{C}_{\rightarrow 1}(x)$ if and only if $x$ is the connected graph formed by adding one extra vertex and one extra edge to a cycle. 
		\item $conn(x)$ holds if and only if $x$ is a connected graph.
	\end{enumerate}
 \end{lem}
 Notice that from the definability of $\mathcal{C}_{\rightarrow 1}(x)$ we also have definability of the graph $C_{j \rightarrow 1}$ which stands for the member of $\mathcal{C}_{\rightarrow 1}$ of order $j+1$ because the family is totally ordered by number of vertices and for similar reasons as Observation \ref{obs:totalOrderDefble}. Additionally, given a parameter $n$, we can obtain $C_{n \rightarrow 1}$.
 Using the basic predicates, we define intermediate predicates helpful in defining $\psi_{opres}$. 
 \begin{lem}[Wires]
 \label{lem:disjointLargeCycles}
	The following predicates are definable in the induced subgraph order:
	\begin{enumerate}
		\item $csum(x,n)$ holds if and only if $n \in \mathcal{N}$ and $x=\bigcup_{i=1}^n C_{n+i+2}$.
		\item $psum(x,n)$ holds if and only if $x=\bigcup_{i=1}^n P_{n+i+1}$.
	\end{enumerate}
 \end{lem}
 \begin{proof}
	We can construct the object
	$\bigcup_{i=1}^n C_{n+i+2}$ as follows:
	\begin{flalign*}
	 	csum(x,n) := & \forall z \; maximalComp(x,z) \; \supset \; \mathcal{C}(z) \; \wedge \\
	 	&  cardCond(x,n) \;  \wedge \; allCycles(x,n)\\
	 	&\text{where}\\
	 	cardCond(x,n):= & \mathcal{N}(n) \; \wedge \; |x| = n^2 + n(n+1)/2 + 3n \\
	 	allCycles(x,n) := &\forall m \in \mathcal{N} \; (n+3 \leq m \leq 2n+2) \; \supset \; C_{m} \leq_i x
	 \end{flalign*}
	 
	 If $x=\bigcup_{i=1}^n C_{n+i+2}$ then it clearly satisfies the formula $csum(n,x)$. Suppose $x$ satisfies $csum(x,n)$. Therefore it contains every cycle $C_m$ for $n+3 \leq m \leq 2n+2$ as induced subgraph. Suppose a copy of $C_{m}$ and $C_{m'}$ ($m \neq m')$) present in $x$ share a common vertex $v$. Consider the subgraph $g$ formed by the vertices of these copies: it is connected. Hence the component of $x$ containing $g$ is not a cycle, which contradicts $\forall z \; maximalComp(x,z) \supset \mathcal{C}(z)$. Hence any copy of $C_m$ is disjoint from a copy of $C_{m'}$ present inside $x$. But the cardinality condition imposed by $cardCond$ implies that there is a unique copy of each $C_m$. In fact, there are no other vertices in $x$ apart from vertices belonging to the cycles $C_m$. There cannot be any edges between different cycles, again because of the $maximalComp$ condition. Hence $x$ is exactly the graph $\bigcup_{i=1}^n C_{n+i+2}$.

	 Next we show how to construct the graph $\bigcup_{i=1}^n P_{n+i+1}$, which is formed by deleting one vertex from each cycle in the graph $\bigcup_{i=1}^n C_{n+i+2}$. 
	\[ psum(x,n) := \exists y \; y=\bigcup_{i=1}^n C_{n+i+2} \; \wedge \; x \lessdot^n_i y \; \wedge \; \forall z \; \mathcal{C}(z) \supset z \nleq_i x\]
	i.e. we get the appropriate graph by enforcing the condition that no cycle is an induced subgraph.
 \end{proof}

 We can now define $\psi_{opres}(y,x)$ which holds if and only if $y$ is an o-presentation of $x$:
 \begin{flalign*}
		\psi_{opres}(y,x) := & \tilde{\mathcal{G}}(y) \; \wedge \; constructFromPaths(y,x)\\
		&where\\
		cardCond(x,y) := & |y|_g= |x|_g^2 + |x|_g(|x|_g+1)/2 + 3|x|_g \\
		constructFromPaths(y,x) := & \exists z \; z = x \cup \bigcup_{i=1}^{|x|} P_{|x|+1+i} \; \wedge \; z \lessdot_i^{|x|} y 
 \end{flalign*}
 The formula $\psi_{opres}$ states that $y$ is an o-presentation of appropriate order and deletion of $|x|$ vertices from $y$ gives the disjoint union of $x$ with paths of cardinality $|x|+2$ to $2|x|+1$. Let $y$ be an o-presentation obtained by addition of $|x|$ new vertices $V_{new}=\{ v_1,v_2,...,v_{|x|}\}$ to $z$. Let $V(z)=V(x) \cup V_P$ where $V_P$ are the vertices of $\bigcup_{i=1}^{|x|} P_{|x|+1+i}$. The cardinality of $y$ is $|x|^2+|x|(|x|+1) + 3|x|$ and hence it must contain each cycle in the set of indicator cycles $\{ C_{|x|+i+2} | 1 \leq i  \leq |x| \}$ as induced subgraph. None of these cycles is present in $z$ as induced subgraph since neither $x$ nor any of the graphs $P_{|x|+1+i}$ contain them. Hence every such cycle has to be created in $y$ through addition of new edges between the newly added vertices $V_{new}$ and $V(z)$. Since the creation of each indicator cycle requires at least one new vertex and the number of new vertices is equal to the number of indicator cycles, the only way to get $\bigcup_{i=1}^{|x|} C_{|x|+i+2}$ as an induced subgraph of $y$ is to add two edges connecting the ends of the path $P_{|x|+i+1}$ to $v_i$, for every $i$. This gives us the graph $y' = x \cup \bigcup_{i=1}^{|x|} C_{|x|+2+i}$ where $V(y')=V(z) \cup V_{new}$. We need to add some more edges between $V_{new}$ and the vertices $V(z)$ in $y'$ to get an o-presentation. But by the properties of o-presentations, there is exactly one more edge between each vertex in $V_{new}$ and the vertices $V(x)$ of $z$. Thus the graph $y$ must be an o-presentation of $x$.  

 Moving on to the last predicate $\psi_{edgeOP}(x,i,j)$ , we first need the following intermediate predicate:

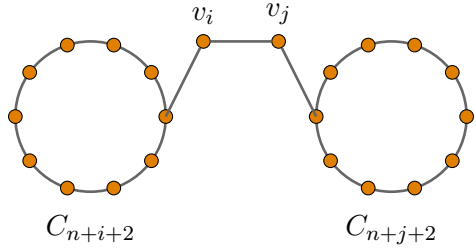
\begin{figure}
 \begin{tikzpicture}
 \def \r {1cm}
 \def \s {0.5}
 \def \d {3.5}
 \def \et {1}
 \tikzstyle{node1} = [draw, circle, fill=c01, scale=\s]
 \tikzstyle{line1}=[draw,line width=\et,color=black!60]
  \newcommand{\polygonC}[3]{
    \node[label=below:#2]{
    \begin{tikzpicture}
     \def \n {#1}
     \foreach \i in {1,...,\n}
     {
     \node[node1] ($\i$) at ({360/\n * (\i  - 1)}:\r) {};
     }
    \end{tikzpicture}
    };
  }
  \draw[line1] (0,0) circle [radius=\r];
  \polygonC{10}{$C_{n+i+2}$};
  \begin{scope}[xshift=4cm]
    \draw[line1] (0,0) circle [radius=\r];
    \polygonC{10}{$C_{n+j+2}$};
  \end{scope}
  \node[node1,label=above:$v_i$] (vi) at (1.5,1) {};
  \node[node1,label=above:$v_j$] (vj) at (2.5,1) {};
  \draw[line1] (1,0)--(vi)--(vj)--(3,0);
  \end{tikzpicture}
  \caption{The $CP_4C$ graph corresponding to an edge between vertices $v_i$ and $v_j$.}
 \label{fig:CP4C}
\end{figure}

 \begin{lem}[Wires \cite{wires2016definability}]
 \label{lem:CP4C}
	$CP_4C(x,i,j)$ if and only if $i,j \in \mathcal{N},3< i < j$ and $x$ is formed by adding to the graph $C_i \cup C_j$ two additional vertices $v_1,v_2$ and the edge $v_1v_2$, one edge between $C_i$ and $v_1$ and one edge between $C_j$ and $v_2$. We denote $x$ by $CP_4C(i,j)$.
 \end{lem}
 \begin{proof} 
 \begin{flalign*}
	CP_4C(x,i,j) := &  conn(x) \; \wedge \; \mathcal{N}(i) \; \wedge \; \mathcal{N}(j) \;  \wedge\; 3<i<j \; \wedge\\
	& |x|=i+j+2 \; \wedge \; C_{i \rightarrow 1} \cup C_j \lessdot_i x \; \wedge \\
	& C_{j \rightarrow 1} \cup C_i \lessdot_i x
 \end{flalign*}
 We show that any $x$ satisfying the above formula is $CP_4C(i,j)$.
 From the definition, $x$ has to be obtained by adding one new vertex $v$ and some number of edges which are incident on $v$ to $g=C_{i \rightarrow 1} \cup C_j$. Let $v_1$ be the unique degree 1 vertex of $g$. Notice that there is only one copy of $C_j$ present as subgraph in $x$ because of cardinality constraints. There must be at least one edge between $v$ and $C_j$ (connectivity constraint). However, if there were multiple edges, we cannot get $C_{j \rightarrow 1}$ as induced subgraph of $x$.
 By the connectivity constraint, we must also have an edge between $v$ and $C_{i \rightarrow 1}$. Suppose there is an edge between $v$ and some vertex $v_2$ in $C_{i \rightarrow 1} $ which is not $v_1$ i.e. $v_2$ is a cycle vertex. Then it is impossible to obtain $C_{j \rightarrow 1} \cup C_i$ from $x$ by deleting a single vertex (since both $v$ and $v_1$ are non-cycle vertices attached to $C_i$ in $x$). Thus there are no such edges $vv_2$. However, there must be an edge $vv_1$ due to the connectivity constraint. Thus the only graph satisfying the formula is $CP_4C(i,j)$.
 \end{proof}

 We can now write 
 \begin{flalign*}
	\psi_{edgeOP}(x,i,j) :=& \exists y \; x \in \tilde{y}\; \wedge \; \exists m \; (|x|=m^2 +m(m+1)/2 +3m) \; \wedge \\
	 	& CP_4C(m+i+2,m+j+2) \leq_i x
 \end{flalign*}

 The existence of an edge between vertices $v_i$ and $v_j$ in the graph $y$ is captured by the presence of a $CP_4C$ induced subgraph in $x$ (which is an o-presentation of $y$) with appropriate parameters and this is stated by the formula $\psi_{edgeOP}$. 

 It remains to define $\tilde{\mathcal{G}}$. To do so, we need some additional basic predicates.
 \begin{lem}[Wires \cite{wires2016definability}]
 \label{lem:basicPredInduced3}
	The following predicates are definable in ($\mathcal{G},\leq_i,P_3$).
	\begin{enumerate}
		\item $z=x \cup y$ if and only if $z$ is the disjoint union of $x$ and $y$.
		\item $C_{\rightarrow 2}(x)$ if and only if $x$ is formed from the graph $g$ which satisfies $C_{\rightarrow 1}(g)$ by adding an additional vertex and joining it to the unique vertex in $g$ which has degree 1.
		\item $pointedCycleSum(x,y,z)$ if and only if $x$ and $y$ are incomparable cycles and $z$ is formed by starting with the graph $x \cup y$ and adding one extra vertex $v$ and two extra edges, one from $v$ to any vertex of $x$ and another from $v$ to any vertex of $y$. We will write $z=x +_p y$ for short. 
	\end{enumerate}
 \end{lem}
 Definability of constants  holds for the family $\mathcal{C}_{\rightarrow 2}$ by observation \ref{obs:totalOrderDefble}.
 Using the above, we can define another intermediate predicate: 
 \begin{lem}
 \label{lem:bicycle}
	The family $bicycle(x)$ if and only if $x$ is formed by adding an edge between two unequal cycles, is definable.
 \end{lem}
 \begin{proof}
 \begin{flalign*}
	bicycle(x) := & conn(x) \; \wedge \; \exists y \; \exists z \; \mathcal{C}(y) \; \wedge \; \mathcal{C}(z) \; \wedge \; y \neq z \\
	& |x|=|y|+|z| \; \wedge \; y \leq_i x \; \wedge \; z \leq_i x \; \wedge \\
	& \forall w \; (w \neq C_{|y| \rightarrow 2} \; \wedge \; C_{|y| \rightarrow 1} \lessdot_i w) \supset w \nleq_i x \\
	& \forall w \; (w \neq C_{|z| \rightarrow 2} \; \wedge \; C_{|z| \rightarrow 1} \lessdot_i w) \supset w \nleq_i x 
 \end{flalign*}
 Since the two cycles $y$ and $z$ are induced subgraphs and the cardinality constraint implies that there are no other vertices apart from the cycle vertices, $x$ is restricted to graphs which are formed by adding edges between $z$ and $y$. There is at least one edge due to the connectedness constraint. We avoid multiple edges by avoiding induced subgraphs which contain two edges incident on either cycle.
 \end{proof}
 Now we can define the set of o-presentations.

 \begin{lem}[Wires \cite{wires2016definability}]
	\label{lem:defineTildegInduced}
	The predicate $\tilde{\mathcal{G}}(x)$ is definable in the induced subgraph order.
 \end{lem} 
 \begin{proof}
 \begin{flalign*}
	\tilde{\mathcal{G}}(x) :=& \exists n \; cardCond(x,n) \; \wedge \; hasC1s(x,n) \; \wedge \\
	& hasUnionOfCycles(x,n)\; \wedge \; noMultiEdge(x,n) \\
	& \wedge \; noPointedCycleSums(x,n)\; \wedge \; noBicycles(x,n)\\
	&\text{where}\\
	cardCond(x,n):= & \mathcal{N}(n) \; \wedge \; |x| = n^2 + n(n+1)/2 + 3n \\
	hasC1s(x,n) := & \forall i \; (1 \leq i \leq n) \; C_{i+n+2 \rightarrow 1} \leq_i x  \\
	hasUnionOfCycles(x,n) :=& \bigcup_{i=1}^n C_{n+i+2} \leq_i x \\
	noMultiEdge(x,n) := & \forall (1 \leq i \leq n) \; \forall z \; (C_{n+i+2 \rightarrow 2} \neq z \; \wedge \;   C_{n+i+2 \rightarrow 1} \lessdot_i z) \\
	& \quad \supset z \nleq_i x\\
	noPointedCycleSums(x,n) := & \forall (1 \leq i < j \leq n) \; C_{n+i+2}+_p C_{n+j+2} \nleq_i \\
	noBicycles(x,n) := & \forall y \; (bicycle(y) \; \wedge \; |y| > 2n) \supset y \nleq_i x 
 \end{flalign*}
 In order to show that a graph $x$ is an o-presentation, we need to show that the vertex set $V$ of $x$ can be partitioned into two sets $V_1$ and $V_2$ such that:
 \begin{enumerate}
	\item Let $|V_2|=n$. The graph $g$ induced on $V_1$ is $\bigcup_{i=1}^n C_{n+i+2}$. We denote the vertex set of each large cycle $C_{n+i+2}$ by $V'_i$ and $\bigcup_i V'_i = V_1$.
	\item There is a bijection $f:\{ V'_1,V'_2,...,V'_n\} \rightarrow V_2$ such that there is an edge from a unique vertex of $V'_i$ to $f(V'_i)$ and there are no other edges between the large cycles in $x$ and the vertex set $V_2$.
 \end{enumerate}

 The formula $cardCond$ states that the graph has as many vertices as required to contain as induced subgraph a graph on $n$ vertices and cycles of order $n+i+2$ for each $i$ between $1$ and $n$. 

 $hasUnionOfCycles$ states that the disjoint union of all the required cycles is an induced subgraph. Because of the cardinality constraint already imposed, this implies that there is a unique copy of each cycle in $x$. Let $V_1$ be the set of vertices which induce the graph $\bigcup_{i=1}^n C_{n+i+2}$. The remaining vertices of $x$ form the set $V_2$, which has cardinality $n$. No restriction is placed on the edges between the non-cycle vertices $V_2$. It remains to place appropriate restrictions on $V_1$ in order to make sure that the resulting graph $x$ is of the required form. The formula $hasC1s$ states that the $C_{\rightarrow 1}$ are induced subgraphs, thus there is at least one edge from every indicator cycle to the rest of the graph. The formula $noMultiEdge$ ensures that there are no multiple edges between a indicator cycle and the rest of the graph while $noPointedCycleSums$ ensures that two different indicator cycles dont have an edge to a vertex $v$ external to the two cycles. At this point, the contraints ensure that there is exactly one edge incident on each indicator cycle which has its other end elsewhere. But this does not rule out the possibility of two cycles directly connected by an edge. To rule this out, we have  $noBicycles$. Note that by design all indicator cycles are of different length. Thus the edge enforced via $hasC1s$ must be between a indicator cycle and a vertex in $V_2$. Together, this implies the existence of the bijection $f$ between indicator cycles and vertices in $V_2$. 	 This ends the proof that ($\mathcal{G},\leq_i,P_3$) satisfies condition $C2$.
 \end{proof}

This completes the construction of explicit formulae required to show that the induced order is capable.

 \begin{lem}[Wires \cite{wires2016definability}]
 \label{lem:indSubIsCapable}
 	The structure ($\mathcal{G},\leq_i,P_3$) is a capable structure.
 \end{lem}
 

\subsection{Subgraph Order}\hspace*{\fill} \\ 
\label{sub:subgraph_order}

	\underline{($\mathcal{G},\leq_s$) satisfies $C1$}:\\
	This has already been proved in \cite{ramanujam2016definability}.

	\underline{($\mathcal{G},\leq_s$) satisfies $C2$}:\\
	We need to show the definability of the formulae $\psi_{opres}$ and $\psi_{edgeOP}$. Similar to the case of the induced subgraph, the result follows from the definability of the intermediate predicates $\tilde{\mathcal{G}}$, $constructFromCycles$ and $CP_4C$. Some of the basic predicates necessary to define three predicates are already known from previous work \cite{ramanujam2016definability}. Two additional graph theoretic relations, namely a ternary relation $disjointUnion(z,x,y)$ which holds if and only if $z$ is the disjoint union of $x$ and $y$ and a binary relation $sameSize(x,y)$ which holds if and only if $x$ and $y$ have the same number of edges are used. 

  First we complete the proof assuming the definability of the three intermediate predicates. This is followed by three subsections where we show the definability of $CP_4C$ in the subgraph order, definability of $\tilde{\mathcal{G}}$ in the structure ($\mathcal{G},\leq_s,sameSize$), definability of $constructFromCycles$ in ($\mathcal{G},\leq_s,disjointUnion,sameSize$). This shows that the structure ($\mathcal{G},\leq_s,disjointUnion,sameSize$) is capable. In the last subsection, we show that the relations $disjointUnion$ and $sameSize$ are definable in the subgraph order, completing the proof that the subgraph order is capable.

 \begin{lem}[\cite{ramanujam2016definability}]
	\label{lem:oldBasicPredSubgraph}
	The following predicates are definable in the subgraph order:
	\begin{enumerate}
		\item $x \lessdot_s y, x \lessdot_{sv} y, x \lessdot_{se} y$ if and only if $y$ is  an upper cover, an upper cover formed by adding one more vertex, an upper cover formed by adding one more edge (respectively) to $x$. 
		\item The family $soc(x)$ if and only if $x$ is a disjoint union of cycles.
 		\item The families $\mathcal{N}, \mathcal{T},\mathcal{K},\mathcal{P},\mathcal{C},\mathcal{S}$ (isolated points, trees, cliques, paths, cycles,  stars respectively).
		\item $conn(x)$ if and only if $x$ is a connected graph.
		
	\end{enumerate}	
 \end{lem}
 For the sake of completeness, the defining formulae for the basic predicates above are given in the appendix.
 Having introduced the necessary basic predicates, we prove the following conditional result:

 \begin{lem}
	\label{lem:condC2Proof}
	The predicates $\psi_{opres}$ and $\psi_{edgeOP}$ are definable in the structure ($\mathcal{G},\leq_s,disjointUnion,sameSize$) assuming the definability of the following predicates:
	\begin{enumerate}
		\item $CP_4C(x,i,j)$ holds if and only if $i,j \in \mathcal{N}, 3< i <j$ and $x$ is constructed from the graph $C_{i \rightarrow 1} \cup C_ {j \rightarrow 1}$ by adding one additional edge between the unique degree one vertices of $C_{i \rightarrow 1}$ and $C_{j \rightarrow 1}$.
		\item $\tilde{\mathcal{G}}(x)$ holds if and only if $x$ is an o-presentation of some graph.
		\item $constructFromCycles(x,y)$ holds if and only if $x$ is constructed by adding $|y|$ edges to the graph $g$ which is the disjoint union of $y$ and all indicator cycles corresponding to $y$.
	\end{enumerate}
 \end{lem}
 \begin{proof}
	\ZEROTWOpsiUNDERSCOREopres
	The proof of correctness of the formula is similar to that in the induced subgraph case. Let $g$ be the graph $y \cup \bigcup_i^{|y|} C_{|y|+i+2}$. The only way to get an o-presentation from the graph $g$ by adding $|y|$ edges is to connect each indicator cycle to a vertex of $y$. The resulting graph has to be an o-presentation of $y$.

	\ZEROTHREEpsiUNDERSCOREedgeOP
	The formula $\psi_{edgeOP}$ we use in the subgraph order is just the formula used in the induced order with $\leq_i$ replaced by $\leq_s$.
	The correctness of the formula follows from the fact that from the construction of o-presentations, it is clear that an indicator cycle (and the $CP_4C$ graphs) is a subgraph of an o-presentation iff it is an induced subgraph of that o-presentation.
 \end{proof}

 We now take up the definability of $CP_4C$, $\tilde{\mathcal{G}}$ and $constructFromCycles$ in that order.

 \subsection{Defining $CP_4C$ in the Subgraph Order}

 We will need the following intermediate predicates:
 \begin{lem}
 \label{lem:predForCP4C}
	The following predicates are definable in the subgraph order:
	\begin{enumerate}
		\item $maximalComp(y,x)$ if and only if $x$ is a maximal component of $y$ under the subgraph order.
		\item $addVert(x,y)$ if and only if $y$ is a connected graph and $x$ is a connected graph formed by adding one additional vertex and one additional edge to $y$.
		\item $\mathcal{C}_{\rightarrow 1}(x)$ if and only if $x$ is the connected graph formed by adding one additional vertex and one additional edge to a cycle.
		\item $\mathcal{C}_{\rightarrow 2}(x)$ if and only if $x$ is a graph which is formed by a taking a graph $g$ from the family $\mathcal{C}_{\rightarrow 1}$ and adding an additional vertex and connecting it to the unique degree one vertex in $g$.
		\item $twoC1s(x,i,j)$ holds if and only if $i,j \in \mathcal{N}$, $3<|i|<|j|$ and $x$ is the graph $C_{i \rightarrow 1} \cup C_{j \rightarrow 1}$. 
	\end{enumerate}
	
 \end{lem}
 \begin{proof}
 The proof of correctness of the defining formulae for the first three predicates is straightforward and essentially follows from their definitions.
	\THREESIXmaximalComp
	\ZEROEIGHTaddVert
 \ONETHREECrightarrowONE

 We need to define a particular constant graph at this juncture. The graph $double3star$ (see Figure \ref{fig:double3star}) has vertex set $V=\{v_1,v_2,v_3,v_4,v_5,v_6 \}$ and edge set\\
 $E=\{ v_1v_2,v_3v_2,v_2v_4,v_4v_5,v_4v_6\}$. The definability of this constant is straighforward and given at the end of Appendix \ref{app:basicPredSubgraph}. 

\begin{figure}
 \begin{tikzpicture}
 \def \r {1cm}
 \def \s {0.5}
 \def \d {3.5}
 \def \et {1}
 \tikzstyle{node1} = [draw, circle, fill=c01, scale=\s]
 \tikzstyle{line1}=[draw,line width=\et,color=black!60]
 
 \node[node1] (v1) at (0,0) {};
 \node[node1] (v2) at (1,0) {};
 \node[node1] (v3) at ({120:\r}) {};
 \node[node1] (v4) at ({2*120:\r}) {};
 \begin{scope}[xshift=1cm]
   \node[node1] (v5) at ({60:\r}) {};
   \node[node1] (v6) at ({300:\r}) {};
 \end{scope}
 \draw[line1] (v5)--(v2)--(v1)--(v3);
 \draw[line1] (v4)--(v1);
 \draw[line1] (v6)--(v2);
 \end{tikzpicture}
  \caption{The graph $double3star$.}
 \label{fig:double3star}
\end{figure}
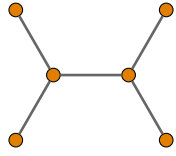

\ONEFOURCrightarrowTWO
 Let $g$ be a graph satisfying the formula $\mathcal{C}_{\rightarrow 2}$. It is a connected graph formed by adding one more vertex and one more edge to a graph $g' \in \mathcal{C}_{\rightarrow 1}$. The formula states that a path of the same order as $g$ is a subgraph of $g$. This implies that we can delete one edge from $g$ to get $P_{|g|}$. There are only two possible such graphs. One of them is the graph we require and the other is the graph $g''$ which is formed from $g'$ as follows. Let $v$ be the vertex of degree three in $g'$ and $u$ be a vertex of degree two which is adjacent to $v$. By deleting the edge $uv$ we can get $P_{|g|}$. But this graph contains $double3star$ as a subgraph, which is disallowed by the formula. The correctness of the formula follows.   
 \ZEROSEVENTWOCONEs
 We now show the correctness of the above formula. Let $x,i,j$ be a tuple satisfying the formula. Every maximal component of $x$ is a graph from the family $\mathcal{C}_{\rightarrow 1}$.
 Since the members of the $\mathcal{C}_{\rightarrow 1}$ family are pairwise incomparable under the subgraph order, any $C_{k \rightarrow 1}$ which is a subgraph of $x$ must be a component of $x$. Hence $C_{i \rightarrow 1}$ and $\mathcal{C}_{j \rightarrow 1}$ are components of $x$. But according to condition $|x|=i+j+2$, there cannot be any other vertices in $x$ except for those belonging to the copy of $C_{i \rightarrow 1}$ or to the copy of $C_{j \rightarrow 1}$. Hence $x$ is exactly the graph $C_{i \rightarrow 1} \cup C_{j \rightarrow 1}$.
 \end{proof}

 Now we can define $CP_4C$:
 \begin{lem}
	\label{lem:definingCP4C}
	The predicate $CP_4C(x,i,j)$ is definable in the subgraph order.
 \end{lem}
 \begin{proof}
  \TWOSEVENCPUNDERSCOREFOURC
  Let $(x,i,j)$ be a tuple satisfying the above formula. We start with $C_i \cup C_j$ and add an edge to get a connected graph. Suppose we connect a vertex from $C_i$ of degree more than 1 to a vertex of $C_j$. The resulting graph contains a graph $g$ as subgraph which satisfies $addVert(g,C_{i \rightarrow 1})$ but does not satisfy $\mathcal{C}_{\rightarrow 2}(g)$. This contradicts the last condition of the formula. Hence the edge has to be added between the unique degree one vertices of $C_i$ and $C_j$.
 \end{proof}

We show the definability of $\tilde{\mathcal{G}}$ using the subgraph order and the $sameSize$ relation in the next subsection. 
 \subsection{Defining $\tilde{\mathcal{G}}$ in the structure ($\mathcal{G},\leq_s,sameSize$)}

 It is possible to define $\tilde{\mathcal{G}}$ in the subgraph order without use of $sameSize$, but the proof becomes messier. We define the necessary intermediate predicates next.
 \begin{lem}
 \label{lem:interPredTildeg}

	The following predicates are definable in the subgraph order:
	\begin{enumerate}
		\item $soc2(x,i,j)$ holds if and only if $x$ is made of the disjoint union of the cycles $C_i$ and $C_j$.
		\item $bicycle(x)$ holds if and only if $x$ is the connected graph formed by adding an edge to a graph $g$ which satisfies $soc2(g,i,j)$ for some values of $i,j$.
		\item $pointedCycleSum(x,i,j)$ holds if and only if $x$ is formed from the graph $g=C_i \cup C_j$ by addition of a vertex $v$, an edge connecting $v$ to $C_i$ and another edge connecting $v$ to $C_j$. We will denote the unique $x$ which is the pointed cycle sum of $C_i$ and $C_j$ by $C_i +_p C_j$.
    \item $csum(x,n)$ holds if and only if $n \in \mathcal{N}$ and $x = \bigcup_{i=1}^n C_{n+i+2}$.
	\end{enumerate}
 \end{lem}
 \begin{proof}
	\FIVETHREEsocTWO
	Let ($x,i,j$) be a tuple satisfying the formula. Then $x$ is a $soc$ on $i+j$ vertices. Suppose $i \neq j$, then $C_i$ and $C_j$ are incomparable. Then it is sufficient to state that both $C_i$ and $C_j$ are subgraphs of $x$, because $C_i \cup C_j$ is also forced as a subgraph of $x$. Along with the cardinality constraint, this implies that $x$ is exactly $C_i \cup C_j$.

	 If $i=j$ and $i=n_1 \times n_2$, then the graph $C_i \cup C_{n_1} \cup ... \cup C_{n_1}$ i.e. the disjoint union of $C_i$ with $n_2$ copies of $C_{n_1}$ also satisfies the conditions encforced so far. Enforcing the condition that every cycle subgraph must be either $C_i$ or $C_j$ disallows this latter graph and the proof of correctness follows.
	\ONEONEbicycle
	The proof of correctness of the above formula follows from the definition.
	\FOURFIVEpointedCycleSum
	Let ($x,i,j$) be a tuple satisfying the formula. It is constructed from the graph $C_i \cup C_j$ by adding one more vertex $v$ and two more edges $e_1,e_2$. Suppose wlog that $e_1$ connects $v$ to one of the cycles $C_i$. Then $e_2$ must join the other cycle $C_j$ to either $v$ or $C_i$. Suppose it connects $C_i$ and $C_j$, then there is a bicycle subgraph of $x$ which is disallowed. Therefore $e_2$ must connect $C_j$ to $v$. Hence the graph $x$ is the required graph.

We can use the same formula as we used in \ref{lem:disjointLargeCycles} with the formula $allCycles'$ being the formula obtained by replacing occurrences of $\leq_i$ in $allCycles$ by $\leq_s$ to define $csum$:
  \TWOEIGHTcsum
  where
  \ZERONINEallCyclesPRIME
  The formula $maximalComp$ is from \ref{lem:predForCP4C} and $cardCond$ is by Observation \ref{obs:cardCond}.
  The formula $allCycles'$ which states that every cycle is present as a subgraph is equivalent to every cycle being present as an induced subgraph under the additional condition that every maximal component is a cycle.    
 \end{proof}

 Along with the predicates defined above, we will use some arithmetic. We note here that relations definable in arithmetic will be translated by use of Corollary \ref{cor:translateArithInduced} into graph formulae. These translated number theoretic relation are not to be confused with arithmetical graph relations; when we talk of using arithmetic in this section, we always mean number theoretic relations.

  By the fact that the subgraph order satisfies condition $C1$,we have:
 \begin{obs}
 \label{obs:cardCond}
	The following predicate is definable in the subgraph order:
	\ONEFIVEcardCond
 \end{obs}
 We now exhibit the defining formula for $\tilde{\mathcal{
 G}}(x)$.

 \begin{lem}
	\label{lem:definingTildegSubgraph}
	The predicate $\tilde{\mathcal{G}}(x)$ is definable in the structure ($\mathcal{G},\leq_s,sameSize$).
 \end{lem}
 \begin{proof}
For any o-presentation $x \in \tilde{g}$, its cardinality is related to the cardinality of  $g$ by $|x|=|g|^2 + |g|(|g|+1)/2 + 3|g|$. In order for an arbitrary $x$ be an o-presentation, it is necessary and sufficient to show that the graph $x$ is formed by adding some number of edges to $g'=\bigcup_{i=1}^{|g|} C_{|g|+i+2 \rightarrow 1}$ such that each new edge is constrained to be between the degree 1 vertices of the $C_{|g|+i+2 \rightarrow 1}$. We will first construct $g'$ given a parameter $|g|$ and use it to define $\tilde{\mathcal{G}}$.

    The predicate $csumHook(x,n)$ which holds if and only if $x = \bigcup_{i=1}^{n} C_{n+i+2 \rightarrow 1}$ is defined by the following formula:
    \begin{flalign*}
      csumHook(x,n) := & \mathcal{N}(n) \; \wedge \; \exists y \; csum(y,n) \; \wedge \; |x|=|y| +n \; \wedge \; ||x||=||y||+n \\
      & \wedge \; y \leq_s x \; \wedge \; \forall i\; (1 \leq i \leq n) \supset \; C_{n+i+2 \rightarrow 1} \leq_s x.
    \end{flalign*}
    The only way to get every graph $C_{n+i+2 \rightarrow 1}$ as a subgraph by adding $n$ vertices and $n$ edges to $\bigcup_{i=1}^n C_{n+i+2}$ is by connecting each new vertex to exactly one of the cycles. 

We can now define $\tilde{\mathcal{G}}$:
 \begin{flalign*}
      \tilde{\mathcal{G}}(x) :=  &\exists n \in \mathcal{N} \; cardCond(x,n) \; \wedge \;  \bigcup_{i=1}^n C_{n+i+2 \rightarrow 1} \leq_s x \\
      & \wedge \; indicatorsInduced(x,n) \; \wedge \; noBicycles(x,n) \\
      & \wedge \; noPointedCycleSum(x,n) \\
      & \text{where} \\
      cardCond(x,y) :=&  |y|= |x|^2 + |x|(|x|+1)/2 + 3|x|\\
      indicatorsInduced(x,n):= & \forall y \; \forall (1 \leq i \leq n) C_{n+i+2 \rightarrow 1} \leq_{se} y \supset \neg(y \leq_s x) \\
      noBicycles(x,n) := & \forall y \; \forall (1 \leq \; i, j \; \leq n) \; (bicycle(y,i,j)  \; \supset y \nleq_s x)\\
      noPointedCycleSum(x,n):= & \forall (1 \; \leq i, j \; \leq n) \; C_{n+i+2} +_p C_{n+j+2} \nleq_s x
\end{flalign*}
 The formula $cardCond$ enforces a cardinality condition that ensures that the every vertex of $x$ has to be used to witness the fact that $\bigcup_{i=1}^n C_{n+i+2 \rightarrow 1}$ is a subgraph. We now need to make sure that none of the extra edges occur and that $\bigcup_{i=1}^n C_{n+i+2 \rightarrow 1}$ occurs as an induced subgraph. The edges inside a particular $C_{n+i+2 \rightarrow 1}$ are forbidden by $indicatorsInduced$, edges between two indicator cycles are forbidden by $noBicycles$ and $noPointedCycleSum$ ensures that two cycles are not connected to the a vertex which represents a degree 1 vertex in a $C_{n+i+2 \rightarrow 1}$. Thus the only extra edges allowed are between the degree 1 vertices of the $C_{n+i+2 \rightarrow 1}$.
 \end{proof}

We take up the third and last of the three predicates required to establish condition $C2$ for the structure ($\mathcal{G},\leq_s,disjointUnion,sameSize$) in the following subsection.
 \subsection{Defining $constructFromCycles$ in the Structure ($\mathcal{G},\leq_s,disjointUnion, sameSize$)} 
 \label{sub:DefiningConstructFromCycles}
 The relation $constructFromCycles(y,x)$ states that $y$ is formed by adding $|x|$ number of edges to the graph obtained by the disjoint union of $x$ with all its indicator cycles.  The construction of the disjoint union of a set of indicator cycles has already been shown in Lemma \ref{lem:interPredTildeg}, but defining arbitrary disjoint union and counting the number of edges of a graph take more work.

 \begin{lem}
	\label{lem:defineConstructFromCyclesConditional}
	The predicate $constructFromCycles$ is definable in \\ ($\mathcal{G},\leq_s,disjointUnion,sameSize$). 
 \end{lem}
 \begin{proof}
 \begin{obs}
	\label{obs:sameSize}
	We are able to do arithmetic operations on the number of edges once the predicate $countEdges$ has been defined. The predicate $countEdges$ enables us to define the predicate $sameSize(x,y)$ which holds if and only if $x$ and $y$ have the same number of edges. We will denote this by $||x||=||y||$.
	\FOUREIGHTsameSize
 \end{obs} 
 From the above observation, it follows that we can write formulae such as $||z||+|x|=||y||$.
 We can now define $constructFromCycles$:
 	\TWOONEconstructFromCycles
 	\end{proof}

 
Using the above lemma and Lemmas \ref{lem:definingCP4C}, \ref{lem:definingTildegSubgraph} and \ref{lem:condC2Proof} we can conclude that ($\mathcal{G},\leq_s,disjointUnion,sameSize$) satisfies condition ($C2$). We have already shown that the subgraph order satisfies conditions ($C1$) and ($C3$) and so we get the following result:

 \begin{lem}
 \label{lem:subDisSamCapable}
   The structure ($\mathcal{G},\leq_s,disjointUnion,sameSize$) is capable.
 \end{lem}

 To complete the proof that the subgraph order is capable, we need to define $disjointUnion$ and $sameSize$ in the subgraph order.

\subsection{Defining $disjointUnion$ and $sameSize$ in the Subgraph Order} 
\label{sub:definingDisSamInSub}

 \subsubsection{Definability of $disjointUnion$}

 The strategy employed to define $disjointUnion$ uses a predicate $maxCopies$ which counts the number of vertex disjoint copies of a graph present as a subgraph in another graph. Using $maxCopies$, we ensure that the multiplicity with which a component occurs in the graph $x \cup y$ is the same as the multiplicity with which it occurs in a given graph $z$.

 \begin{lem}
 \label{lem:defineMultCopies}
 	The following predicates are definable in the subgraph order:
 	\begin{enumerate}
 	 	\item $mult(x,y,n)$ holds if and only if $y$ is a connected graph and $x$ is the disjoint union of $n$ many components, each of which is $y$. 
 	 	\item $maxCopies(x,y,n)$ holds if and only if $y$ is connected and the graph $g$ satisfying $mult(g,y,n)$ is a subgraph of $x$ but not the graph $g'$ satisfying $mult(g',y,n+1)$.
 	 \end{enumerate}
 \end{lem}
 We will refer to the unique $n$ satisfying $maxCopies(x,y,n)$ (where $x$ and $y$ are fixed graphs) as the ``the number of copies of $y$ in $x$'' and will just write $maxCopies(x,y)$ to denote $n$.
 \begin{proof}
 	First we define the relation $mult(x,y)$ which holds if and only if $x$ is made of up some arbitrary number of components, each of which is $y$.
 	\THREEEIGHTmult
 	where
  \[ zero(x) := x = \emptyset_g \]
 	\THREESEVENmaximumComp
 	\FIVEFIVEuniqueCompCard
 	\THREEONEedgeMaximal
  The subformula $zero$ is necessary for the boundary case where $x=\emptyset_g$ which is intended to be the case where $x$ is made of zero number of copies of $y$. 
 	The rest of the defining formula for $mult(x,y)$ enforces three conditions: 
 	\begin{enumerate}
 	 \item $maximumComp(x,y)$: $y$ is a maximum component of $x$ i.e. for any component $c$ of $x$, it is the case that $c \leq_s y$.
 	 \item $uniqueCompCard(x,y)$ : Adding any edge to $x$ gives a graph $z$ such that every maximal component of $z$ has cardinality either $|y|$ or $2|y|$.
 	 \item $edgeMaximal(x,y)$: Adding any edge to $x$ ensures that $y$ is no longer the maximum component of the new graph.
 	\end{enumerate}
 	The proof of correctness of the formulae enforcing these three conditions is straightforward. We now show that they suffice to define $mult$.
 	Let ($x,y$) satisfy the defining formula. We have to show that $x$ is the disjoint union of some arbitrary number of $y$. The predicate $maximumComp(x,y)$  implies that every component of $x$ is a subgraph of $y$. Thus $y$ occurs as a component of $x$ and for any component $c$ of $x$, $c$ is a subgraph of $y$. 

 	Suppose $x$ has a component $c$ which is not $y$. Consider $y \cup c$ which is a subgraph of $x$. Adding an edge between $y$ and $c$ gives us an edge cover $z$ of $x$ which has a component $c'$ of order $|y|+|c|$. Since $y$ is the maximum component of $x$, $c'$ is a maximal component of $z$. Hence by $uniqueCompOrder$ which states that any edge cover of $x$ has the property that every maximal component has order $|y|$ or $2|y|$, we get $|c|=|y|$. But we chose an arbitrary $c$, so all components in $x$ have order $|y|$. Thus any component $c$ of $x$ is obtained by deletion of some number of edges from $y$ while making sure the resulting graph is still connected.

 	Let $c$ be a component which is a strict subgraph of $y$. We can add an edge to $c$ to get $c'$ which is still a subgraph of $y$. The edge cover of $x$ thus formed still has as its maximum component $y$, which contradicts the formula $edgeMaximal$. Hence $c$ is not a strict subgraph of $y$. Hence all components in $x$ are in fact, $y$.

 	Using $mult(x,y)$ we can define the predicate $mult(x,y,n)$ which holds if and only if $x$ is made of exactly $n$ many $y$. This is because we can use arithmetic to get $n=\frac{|x|}{|y|}$.

 	We can now define $maxCopies$:
 	\TWOZEROcopies
 \end{proof}
 We are now ready to define disjoint union.
 \begin{lem}
	\label{lem:disjointUnion}
	The predicate $disjointUnion(z,x,y)$ which holds if and only if $z=x \cup y$ is definable in the subgraph order.
 \end{lem}
 \begin{proof}
	\TWONINEdisjointUnion
	where
	\ONESEVENcompUnionInZ
	\ONEEIGHTcompZInUnion

	Consider any component $c$ of $x \cup y$. Such a $c$ must be a subgraph of either $x$ or $y$ since any component of $x \cup y$ must be a component of either $x$ or $y$ (or both). Hence by $compUnionInZ$, the number of copies of $c$ in $x \cup y$ which is the same as the sum of the number of copies of $c$ in $x$ and $y$, is the number of copies of $c$ in $z$.\\
	By a similar arguement, $compZInUnion$ enforces that the number of copies of a component $c$ in $z$ is the sum of number of its copies in $x$ and $y$ i.e. the number of copies in $x \cup y$.
	We now observe that:
	
  \begin{obs}
    \label{obs:copiesIso}
    For any two graphs $g_1$ and $g_2$, $g_1=g_2$ iff for any component $c$ of $g_1$, $maxCopies(g_1,c)=maxCopies(g_2,c)$ and for any component $c'$ of $g_2$, $maxCopies(g_1,c')=maxCopies(g_2,c')$. 
  \end{obs}
  The forward direction is obvious, we consider the reverse. Let $g_1 \neq g_2$ be two graphs of smallest order which form a counterexample. Let $c$ be a maximal component of $g_1$, then it must also be a maximal component of $g_2$; otherwise either there exists a component $c'$ of $g_2$ which is a supergraph of $c$ or $c$ is not a subgraph of $g_2$. In the first case $maxCopies(g_2,c') >0$ but $maxCopies(g_1,c') =0$, in the second $maxCopies(g_2,c)=0$ but $maxCopies(g_1,c) >0$. Further, $maxCopies(g_1,c)=maxCopies(g_2,c)$ by assumption and let this number be $n_c$. Clearly $g_1=n_cc \cup g_1'$ and $g_2=n_cc \cup g_2'$ i.e. $g_1$ is the disjoint union of a graph $g_1'$ which does not contain $c$ as subgraph with the graph $n_cc$ which is the disjoint union of $n_c$ copies of $c$. Similarly for $g_2$. This implies that $g_1' \neq g_2'$ which contradicts the assumption that $g_1,g_2$ form the smallest counterexample and the observation follows.

	We have shown above that $x \cup y$ and $z$ satisfy the property in the above observation and hence $z=x \cup y$.  This concludes the proof of definability of disjoint union.
 \end{proof}

 The definability of disjoint union immediately gives us the definability of another natural graph theoretic predicate:
 \begin{cor}
 	\label{cor:definingComp}
 	The predicate $comp(y,x)$ which holds if and only if $x$ is a component of $y$ is definable in the subgraph order.
 	 \ONESIXcomp
 \end{cor}

 We take up the definability of $sameSize$ in the subgraph order next.

 \subsubsection{Definability of $sameSize$}
 By Observation \ref{obs:sameSize}, it is sufficient to define the predicate $countEdges(x,n)$ which holds if and only if $x$ has $n$ many edges. We will break this problem up into two parts: counting the number $n$ of edges of a connected graph $x$  using the predicate $countEdgesConn(x,n)$, and counting the number $n$ of components of a graph $x$ using the predicate $countComps(x,n)$. 

 Let a graph $g$ have $m_2$ many components. The minimum number of edges to be added to $g$ to get a connected graph $g'$ is $m_2-1$. It is clear that the number of edges in $g,g'$ is related by $||g|| + m_2 -1=||g'||$. Thus using $countEdgesConn$ and $countComps$ we can define the required predicate $countEdges$.
 \begin{lem}
 	\label{lem:countEdgesConn}
 	The predicate $countEdgesConn(x,n)$ if and only if $x$ is a connected graph and $n$ is the number of edges of $x$ is definable in the subgraph order.
 \end{lem}
 \begin{proof}
 	We construct a gadget to define $countEdgesConn(x,n)$. For any given graph $g$, we construct a graph $g'$ which has the following properties:
  \begin{enumerate}
    \item $g$ is a component of $g'$. 
   	\item Every component of $g'$ is formed by adding some number of edges to $g$.
   	\item For every component $c$ of $g'$, there is a component $c'$ of $g'$ such that $c'$ is an edge-cover of $c$, or it is the case that $c$ is a clique.
   	\item $g'$ is a minimal element under the subgraph order of the set of graphs satisfying the above three properties.
   \end{enumerate}
   It is easy to see that any such $g' = c_0 \cup c_1 \cup_2 ... c_m$ where $c_0=g, c_m=K_{|g|}$ and for every $i>0$, $c_i$ is an edge-cover of $c_{i-1}$ (see Figure \ref{fig:countEdgesConnGadget} for an example). Hence using some arithmetic, we can relate the cardinality of $g'$ to the number of edges of $g$ and retrieve the latter : ${|g| \choose 2} - ||g||+1= |g'|$. Having given the proof idea, we give the formulae below:
  
  First we have a formula which enforces the first three conditions on $g'$.
  \TWOSIXcountEdgesGadgetPRIME
  The condition of minimality is enforced:
  \TWOFIVEcountEdgesGadget
  We use the constructed gadget to do the edge counting:
  \TWOFOURcountEdgesConn
 \end{proof}

 
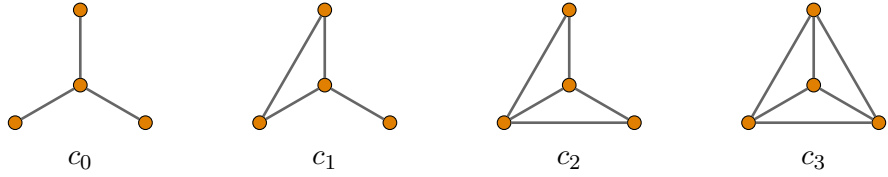
\begin{figure}
 \begin{tikzpicture}
 \def \r {1cm}
 \def \s {0.5}
 \def \d {3.5}
 \def \et {1}
 \tikzstyle{node1} = [draw, circle, fill=c01, scale=\s]
 \tikzstyle{line1}=[draw,line width=\et,color=black!60]
  \foreach \i in {0,1,2,3}
  {
  \begin{scope}[xshift= 2.25*\i cm] 
    \node[node1] (v1) at (\i,0) {};
    \node[node1] (v2) at ($(v1)+(0,1)$) {};
    \node[node1] (v3) at ($(v1)+(210:1)$) {};
    \node[node1] (v4) at ($(v1)+(330:1)$) {};
    \draw[line1] (v2)--(v1)--(v3);
    \draw[line1] (v1)--(v4);
    \node at ($(v1)-(0,1)$) {$c_{\i}$};
    \ifthenelse{\i=1}
     {
       \draw[line1] (v2)--(v3);
     } 
     { \ifthenelse{\i=2} 
         {\draw[line1] (v2)--(v3)--(v4); }
       {
          \ifthenelse{\i=3}
          {
          \draw[line1] (v2)--(v3)--(v4)--(v2);
          }
          {}

       }
     }
  \end{scope}  
  }

  \end{tikzpicture}
  \caption{The gadget $g'$ for counting number of edges of the connected graph $g=S_4$. Note that $g'=c_0 \cup c_1 \cup c_2 \cup c_3$ with $c_0=S_4$ and $c_3=K_4$}
 \label{fig:countEdgesConnGadget}
\end{figure} 
 The other predicate needed to define $countEdges$ is $countComps$.
 \begin{lem}
 	\label{lem:countComps}
 	The predicate $countComps(x,n)$ which holds if and only if $x$ contains $n$ components is definable in the subgraph order.
 \end{lem}
 \begin{proof}
   We show that for any graph $g$, we can construct the graph $g'$ which is obtained by adding all the edges $uv$ such that $u,v$ belong to the same component of $g$ (see Figure \ref{fig:unionOfCliques} for an example). Such a graph is a disjoint union of cliques and we will call this family of graphs $unionOfCliques$. The number of components in both $g$ and $g'$ is the same. We then show how to use arithmetic to count components assuming that the input graph is always from the family $unionOfCliques$.

    \begin{defi}[Extend to Cliques]
  	\label{def:mapExtendToCliques}
  	The map $f:\mathcal{G} \rightarrow \mathcal{G}$
  takes a graph $g = c_0 \cup c_1 \cdots \cup c_n$ (where each $c_i$ is a component) to the graph $f(g)=K_{|c_0|} \cup K_{|c_1|} \cup \cdots \cup K_{|c_m|}$.
  \end{defi}

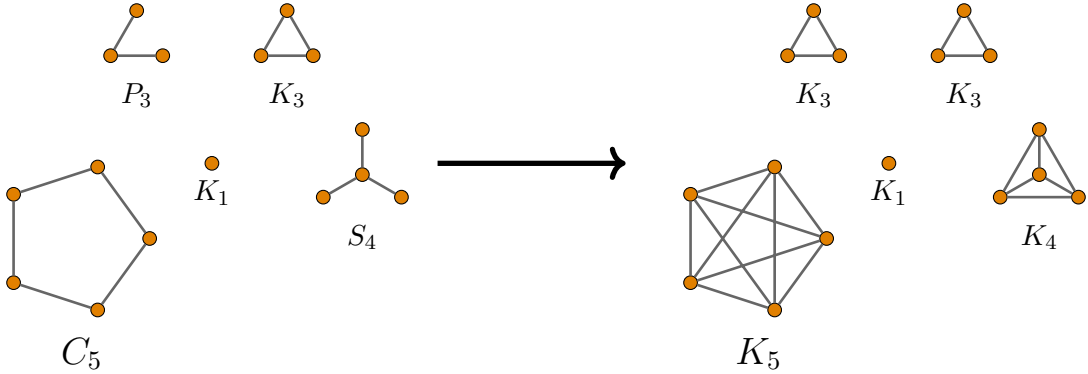
\begin{figure}
 \begin{tikzpicture}
 \def \r {1cm}
 \def \s {0.5}
 \def \d {2cm}
 \def \et {1}
 \def \el {0.4}
 \tikzstyle{node1} = [draw, circle, fill=c01, scale=\s]
 \tikzstyle{line1}=[draw,line width=\et,color=black!60]
  \node[node1,label=below:$K_1$] at (0,0) {};
  \node[shift={(120:\d)},label=below:$P_3$]{
  \begin{tikzpicture}
    \node[] (v1) at (0,0) {};
    \node[node1] (v2) at ($(v1)+(0,\el)$) {};
    \node[node1] (v3) at ($(v1)+(210:\el)$) {};
    \node[node1] (v4) at ($(v1)+(330:\el)$) {};
    \draw[line1] (v2)--(v3)--(v4);
  \end{tikzpicture}
  };
  \node[shift={(60:\d)},label=below:$K_3$]{
  \begin{tikzpicture}
    \node[] (v1) at (0,0) {};
    \node[node1] (v2) at ($(v1)+(0,\el)$) {};
    \node[node1] (v3) at ($(v1)+(210:\el)$) {};
    \node[node1] (v4) at ($(v1)+(330:\el)$) {};
    \draw[line1] (v2)--(v3)--(v4)--(v2);
  \end{tikzpicture}
  };
 \node[shift={(\d,0)},label=below:$S_4$]{
  \begin{tikzpicture}
   \node[node1] (v1) at (0,0) {};
   \node[node1] (v2) at ($(v1)+(0,1.5*\el)$) {};
   \node[node1] (v3) at ($(v1)+(210:1.5*\el)$) {};
   \node[node1] (v4) at ($(v1)+(330:1.5*\el)$) {};
   \draw[line1] (v2)--(v1)--(v3);
   \draw[line1] (v1)--(v4);
  \end{tikzpicture}
 };
   \node[shift={(210:\d)},label=below:\Large $C_5$]{
    \begin{tikzpicture}
    \node[node1] (v11) at ({72*0}:\r) {};
    \node[node1] (v2) at ({72*1}:\r) {};
    \node[node1] (v3) at ({72*2}:\r) {};
    \node[node1] (v4) at ({72*3}:\r) {};
    \node[node1] (v5) at ({72*4}:\r) {};
    \draw[line1] (v11) -- (v2) --(v3)--(v4)--(v5)--(v11);
    \end{tikzpicture}
    };
 \draw[->,line width=2] (3,0) -- (5.5,0);

 \begin{scope}[xshift=9cm]
  \node[node1,label=below:$K_1$] at (0,0) {};
  \node[shift={(120:\d)},label=below:$K_3$]{
  \begin{tikzpicture}
    \node[] (v1) at (0,0) {};
    \node[node1] (v2) at ($(v1)+(0,\el)$) {};
    \node[node1] (v3) at ($(v1)+(210:\el)$) {};
    \node[node1] (v4) at ($(v1)+(330:\el)$) {};
    \draw[line1] (v2)--(v3)--(v4)--(v2);
  \end{tikzpicture}
  };
  \node[shift={(60:\d)},label=below:$K_3$]{
  \begin{tikzpicture}
    \node[] (v1) at (0,0) {};
    \node[node1] (v2) at ($(v1)+(0,\el)$) {};
    \node[node1] (v3) at ($(v1)+(210:\el)$) {};
    \node[node1] (v4) at ($(v1)+(330:\el)$) {};
    \draw[line1] (v2)--(v3)--(v4)--(v2);
  \end{tikzpicture}
  };
 \node[shift={(\d,0)},label=below:$K_4$]{
  \begin{tikzpicture}
   \node[node1] (v1) at (0,0) {};
   \node[node1] (v2) at ($(v1)+(0,1.5*\el)$) {};
   \node[node1] (v3) at ($(v1)+(210:1.5*\el)$) {};
   \node[node1] (v4) at ($(v1)+(330:1.5*\el)$) {};
   \draw[line1] (v2)--(v3)--(v4)--(v2);
   \draw[line1] (v2)--(v1)--(v3);
   \draw[line1] (v1)--(v4);
  \end{tikzpicture}
 };
   \node[shift={(210:\d)},label=below:\Large $K_5$]{
    \begin{tikzpicture}
    \node[node1] (v1) at ({72*0}:\r) {};
    \node[node1] (v2) at ({72*1}:\r) {};
    \node[node1] (v3) at ({72*2}:\r) {};
    \node[node1] (v4) at ({72*3}:\r) {};
    \node[node1] (v5) at ({72*4}:\r) {};
    \draw[line1] (v1) -- (v2) --(v3)--(v4)--(v5)--(v1)--(v3)--(v5)--(v2)--(v4)--(v1);
    \end{tikzpicture}
    };
 \end{scope}
 \end{tikzpicture}
  \caption{The map $f$ taking the graph $g=K_1 \cup P_3 \cup K_3 \cup S_4 \cup C_5$ to the graph $f(g)=K_1 \cup K_3 \cup K_3 \cup K_4 \cup . K_5$ belonging to the family $unionOfCliques$.}
 \label{fig:unionOfCliques}
\end{figure}
  It is easy to define the family $unionOfCliques$:
  \FIVEFOURunionOfCliques
  The predicate $extendToCliques(y,x)$ holds if and only if $y=f(x)$ and is defined as
  \THREETWOextendToCliques
  We show the correctness of the defining formula. The formula states that $y$ is a member of $unionOfCliques$ which is a minimal element under the subgraph order of the set $S$ of all graphs which are supergraphs of $x$ and also belong to $unionOfCliques$.
  Clearly the graph $y=f(x)$ belongs to $S$. We claim that it is the unique minimal element of $S$.
  Consider any $y' \in S$. Since $y'$ is a supergraph of $x=c_0 \cup c_1 \cup ... \cup c_n$, there are vertex sets $V_0 ,V_1,...,V_n$ such that $V_i \subset V(y'), V_i \cap V_j = \emptyset$ for all $i\neq j$ and $x[V_i]$ is a supergraph of $c_i$ with $|V_i|=|c_i|$. Thus $x[V_i]$ is also a connected graph. But this implies that $x[V_i]$ is equal to $K_{|c_i|}$. Hence $f(x)$ is also a subgraph of $y'$ and hence is the unique minimal element of $S$.

  Next we show that through the use of arithmetic, we can extract the number of components from a graph of the family $unionOfCliques$. 

  Arithmetic allows us to create sequences of numbers which can be stored as a single number and manipulate this sequence. 
  \begin{obs}
  \label{obs:sequenceArith}
	The following predicates are definable in arithmetic:
	\begin{enumerate}
		\item $\phi_{sequence}(n,i,j)$ holds if and only if the largest power of the $i^{th}$ prime which divides $n$ is $j$. We will abuse notation in the sequel and refer to the number $2^{n_1} \times 3^{n_2} \times ... \times p_k^{n_k}$ as the sequence ($n_1,n_2,...,n_k$).
		\item $\phi_{sequenceSum}(n,m)$ holds if and only if the sum of the exponents of all primes dividing $n$ is $m$ i.e. $n=2^{n_1} \times 3^{n_2} \times ... \times p_k^{n_k}$ and $m=n_1+n_2+...+n_k$.
	\end{enumerate}
  \end{obs}
  The intended usage of $\phi_{sequence}(n,i,j)$ is that $n$ is a sequence of numbers and we can extract the $i^{th}$ item in the sequence if we are given the index $i$.

  Our objective is to construct the sequence $l_x$ containing the number of times each clique occurs as a component in a given member $x$ of $unionOfCliques$. For example, if $x = K_2 \cup K_2 \cup K_4 \cup K_5$ then the sequence $l$ corresponding to it would be ($0,2,0,1,1$) since there are no $K_1$ components, two $K_2$ components etc. The strategy employed to construct $l_x$ is to first construct the related sequence $l'_x$ whose $i^{th}$ entry takes the value $maxCopies(K_i,x)$. The sequence $l'_x$ is easier to construct and in the case of the example considered above takes the value ($13,6,2,2,1$).

  In the general case, we have the following definition which relates these two sequences:
  \begin{defi}
  \label{defi:convertsequence}
	Let $A \subseteq \mathbb{N}^*$ be the subset of sequences of natural numbers such that for any $l \in A, l=(n_1,n_2,...,n_k)$, it is the case that $n_k >0$. \\
	We define a map $f_1:A \rightarrow \mathbb{N}^*$ as follows:\\
	For $l=(n_1,n_2,...,n_k) \in A$, $f(l)=l' = (m_1,m_2,...,m_k)$ where
	\[ m_i= \Sigma_{j=1}^k \lfloor \frac{j}{i} \rfloor \times n_j \]
	Let $B \subseteq \mathbb{N}^*$ be the image set under the map $f_1$. We can define the inverse map $f_1^{-1} : B \rightarrow A$ of the map $f_1$. Given $l'=(m_1,m_2,...,m_k)$, it is possible to compute $n_i$ if we know the value of $m_i,n_{i+1},n_{i+2},...,n_k$:   
	\begin{flalign*}
		n_k= & m_k\\
		n_i = & m_i - \Sigma_{j=i+1}^k \lfloor \frac{j}{i} \rfloor \times n_j
	\end{flalign*}
  \end{defi}
  From the above definition, it is clear that the inverse map $f_1^{-1}$ is well defined and is recursive and hence, definable in arithmetic by Theorem \ref{thm:recInArith}.
  \begin{obs}
	\label{obs:sequenceConvert}
	The predicate $\phi_{sequenceConvert}(m,n)$ holds if and only if $n,m$ are sequences such that $n=f_1^{-1}(m)$ is definable in arithmetic.
  \end{obs}

  The predicate $makeSequenceFromUOC(x,n)$ which holds if and only if $n$ is a sequence, $x \in unionOfCliques$ and for every $i$, the $i^{th}$ member of the sequence $n$ is the number $maxCopies(x,K_i)$; is definable in the subgraph order:
  \THREEFIVEmakeSequenceFromUOC
  Note that $\psi^{trans}_{sequence}$ is the appropriate translated formula by Corollary \ref{cor:translateArithInduced} defining a number theoretic predicate.

  We can define $countComps$ by creating the sequence of number of copies, translating it under the map $f_1^{-1}$ and adding up the elements of the latter sequence:
  \TWOTWOcountComps
 \end{proof}
 This concludes the proof of Lemma \ref{lem:countComps}.
 We are finally able to define $countEdges$ which is required for defining $constructFromCycles$:
 \begin{lem}
 	\label{lem:countEdges}
 	The predicate $countEdges(x,n)$ holds if and only if $x$ has $n$ edges is definable in the subgraph order.
 \end{lem}
 \begin{proof}
  Using Lemma \ref{lem:countEdgesConn} (definability of $countEdgesConn$) and Lemma \ref{lem:countComps} (definability of $countComps$), we can define $countEdges$: 
  \TWOTHREEcountEdges
 \end{proof}


Using Lemmas \ref{lem:subDisSamCapable}, \ref{lem:countEdges} and \ref{lem:disjointUnion}, we have the following result:

 \begin{lem}
 	\label{lem:subgraphCapable}
 	The structure ($\mathcal{G},\leq_s$) is capable. 
 \end{lem}
  This concludes the major technical contribution of this paper.



\subsection{Minor Order} 
\label{sub:minor_order}
 The following lemma implies that ($\mathcal{G},\leq_m,sameSize$) is a capable structure.

 \begin{lem}
	 \label{lem:definingSubgraphInMinor}
	 The subgraph order is definable in the structure ($\mathcal{G}, \leq_m, sameSize$).
 \end{lem}
 \begin{proof}
 	\[ x \leq_s y := \exists z \; x \leq_m  z \; \wedge \; sameSize(x,z) \; \wedge \; z \leq_m y \; \wedge \; |z|=|y|\]
 	Suppose $x \leq_s y$. Then $y$ can be constructed from $x$ in two steps $s1,s2$. Step $s1$ involves addition of an arbitrary number of vertices to give a graph $z$. Step $s2$ involves addition of an arbitrary number of edges to the graph $z$ to get $y$. Note that this two step construction characterizes the subgraph order. The formula captures this two step construction and its correctness follows from the following observation:
 	\begin{obs}
 		Let $g_1,g_2 \in \mathcal{G}$. If $g_1,g_2$ have the same number of edges or same cardinality, then
 		\[ g_1 \leq_s g_2 \iff g_1 \leq_m g_2.\]
 	\end{obs}
 	This is due to the fact that a contraction operation decreases both the number of edges as well as the number of vertices in a graph. Thus constraining either one of these parameters to remain constant means that contraction operations cannot be used.
 \end{proof}
 This concludes the proof of Lemma \ref{lem:definingSubgraphInMinor}.
We conclude this section with the formal statement of our main result: 
 \begin{thm}[Maximal Definability Property of Graph Orders]
	\label{thm:graphOrdersRecPred}
  The structures $(\mathcal{G},\leq_s)$ ,$(\mathcal{G},\leq_i,P_3)$ and $(\mathcal{G},\leq_m,sameSize)$ have the maximal definability property. In other words, given a structure $\mathcal{A}$, let $Def(\mathcal{A})$ be the set of all  predicates definable in $\mathcal{A}$. Then
	\[ Def(\mathcal{G},plus_t,times_t) = Def(\mathcal{G},\leq_s) = Def(\mathcal{G},\leq_i,P_3) = Def(\mathcal{G},\leq_m,sameSize)\]
 \end{thm}

   
 \begin{cor}
 \label{cor:recPredInGraphOrder}
  Each of the structures $(\mathcal{G},\leq_s)$, $(\mathcal{G},\leq_i,P_3)$ and $(\mathcal{G},\leq_m,sameSize)$ can define any recursive predicate over graphs.
 \end{cor}

\section{Discussion}
\label{sec:discussion}
 Definition \ref{def:capableStructure} of a capable structure  places emphasis on the fact that the family $\mathcal{N}$ is used to represent the natural numbers (in condition $C1$, definability of arithmetic). The use of $\mathcal{N}$ to represent numbers makes the definition of cardinality (condition $C3$) straightforward in the case of graph orders because of the following equivalence which holds when $\leq$ is one of $\leq_i,\leq_s$ and $\leq_m$:
 \[ (|x|=|y|) \iff [\forall z \; \mathcal{N}(z) \; \supset (z \leq x \iff z \leq y)]\]
 The above equivalence does not help us when we consider structures over graphs which are not graph orders; it may be more appropriate to use a graph family different from $\mathcal{N}$ to represent numbers, together with a different strategy for defining cardinality. 
     
  We may also be interested in considering structures over a domain $\mathcal{G}_0 \subset \mathcal{G}$. For instance, we could consider studying the set of all finite trees under the subtree order (the restriction of subgraph to trees), denoted ($\mathcal{T},\leq_{st}$). We can ask the question of whether every recursive predicate over trees is definable in this structure. Clearly, $\mathcal{N}$ is not contained in $\mathcal{T}$ and hence can not be used as the basis for arithmetic. 






 Our theorems as stated use a particular encoding $U\!N$ of graphs as strings (equivalently, as numbers). Suppose there is a Turing machine $M'$ which uses a different encoding $U\!N'$ of graphs as strings. As long as the map $f:U\!N'(\mathcal{G}) \rightarrow U\!N(\mathcal{G})$ defined as $f(U\!N'(g))= U\!N(g)$ which takes us between these two encodings is also recursive, we can still obtain the required formula in graph orders. This is because we can obtain a Turing Machine $M$ which corresponds to the same graph language as $M'$ but uses the encoding $U\!N$. Hence our results are \textit{encoding independent} (upto recursive transformations). Similarly, the total order $\leq_t$ which is used to define arithmetical predicates over graphs can be replaced by any other total order $\leq'_t$ as long as $\leq'_t$ is arithmetical with respect to $\leq_t$. 

 One of the orders for which we are unable to obtain the main result is the minor order. The construction of ``o-presentations'' depends on the existence of an infinite family of incomparable elements in the order (in our case, the family $\mathcal{C}$ of cycles was used). However by the Graph Minor Theorem of Robertson and Seymour \cite{robertson2004graph}, no infinite antichain exists in the minor order. An indirect way of achieving our goal would be to define the $sameSize$ predicate in the minor order. It is not clear how one can circumvent either problem and this is a question we would like to address in future. On the other hand we do not have the tools to tackle the problem of proving inexpressibility in such rich structures. 

 There are subtle differences in the definability in the subgraph and induced subgraph orders. For instance, the predicate $\mathcal{N}(x)$ which identifies the graphs that consist of isolated vertices, is easily defined in both the subgraph and induced subgraph orders. However, the predicate $|x|=|y|$ which states that the cardinality of the graphs $x$ and $y$ is the same, is easy in the subgraph order but seems to take a lot of work in the induced subgraph order. Our main result shows that the subgraph and induced subgraph orders define the same set of graph relations and can define each other. We would like to understand the logical resources required to define the subgraph order in the induced subgraph order; for instance, is it definable using an existential formula with constants? We would also like to understand the comparative strength of definability in these two orders when the same syntactic fragment is considered : are the predicates definable in the three variable fragment $FO^3(\mathcal{G},\leq_s)$ the same as those definable in $FO^3(\mathcal{G},\leq_i,P_3)$?

 
 We do not say anything in this paper about the homomorphism order $\leq_h$. The structure ($\mathcal{G},\leq_h$) is a preorder and not a partial order. There are infinitely many  equivalence classes $\equiv_h$ (defined by $x \leq_h y \wedge y \leq_h x$) and thus we will at least need to expand the vocabulary by adding infinitely many constants, one for each equivalence class. The most natural candidate for the set of constants is the set of cores $\mathscr{C}_h$ which are unique minimum elements of the equivalence classes. Even after this, it is far from clear that we can define every element of $\mathcal{G}$ in ($\mathcal{G},\mathscr{C}_h,\leq_h$). In addition, some equivalence classes collapse important families of graphs. For example, all bipartite graphs are homomorphic to $K_2$. This implies that any forest belongs to $[K_2]_{\equiv h}$ and much of the gadget construction in \cite{ramanujam2016definability} using trees is rendered useless. 

 The resulting quotiented structure ($\mathcal{G} / \equiv_h,\leq_h$) (which is isomorphic to ($\mathscr{C}_h,\leq_h$)) is in some sense, a richer object than the orders considered in this paper. Every countable poset is embeddable in ($\mathscr{C},\leq_h$) \cite{hell2004graphs}; this is clearly not the case with $\leq_i,\leq_s$ and $\leq_m$ each of which is well-founded and thus does not contain an infinite descending chain. Given the undecidability result of Hatami \cite{hatami2011undecidability}, one expects the first order theory of this structure to be undecidable. In fact, we expect:
 \begin{conj}
	\label{conj:arithInHomOrder}
	The first order theory of ($\mathbb{N},\phi_{+}, \phi_{\times}$) (where $\phi_{+}, \phi_{\times}$ are predicate versions of $+$ and $\times$) is definable in ($\mathscr{C}_h,\leq_h$).
 \end{conj}

  This is an example of a structure over graphs which is very rich but is \textit{incapable of accessing the internal structure} of graphs and thus is not able to define all recursive predicates over graphs.

 We may also relate definability in graph orders to Descriptive Complexity (DC) classes. One would ideally like to come up with a `nice' fragment of graph order which characterizes a given DC class. The simplest formulae in graph order are of the form $g \leq g'$ and are $\NP-$complete when seen as computational problems with inputs $g,g'$. This gives us a hint that in order to characterize smaller complexity classes such as the set of polynomial time relations over graphs, we need a vocabulary containing graph relations of lesser complexity than graph order. It is possible that larger complexity classes such as $ELEMENTARY$, which corresponds to higher order logic, can be captured as a `nice' fragment of graph order. 

A fact which distinguishes complexity classes smaller than $\P$ from those bigger than $\NP$ is that the characterizations as DC classes in the former case have been obtained only when finite ordered graphs (aka finite labelled graphs according to our Definition \ref{defi:labelledGraph}) are considered; whereas there exist characterizations for $\NP$ and larger complexity classes for finite graphs without an order \cite{immerman2012descriptive}. Characterizing the polynomial time predicates over finite graphs not equipped with an order remains a major open problem \cite{grohe2008quest}. The notion of an ordered graph can be captured in graph order using o-presentations and thus a set of ordered graphs can also been seen as a set of (unordered) graphs; but this latter set depends on the definition of an o-presentation and thus cannot be seen as natural. One way to attack this problem is to consider the set $\mathcal{G}^o$ of finite ordered graphs with corresponding structures over ordered graphs and search for fragments of such theories which correspond to computational complexity classes over ordered graphs.


 The way we have arrived at our result is very roundabout since it uses the power of arithmetic inside graph orders. We would like to identify ``natural computational predicates'' over graphs which are the equivalent of the $bit$ predicate and exponentiation in arithmetic. In addition, we note that the method of computation we use essentially puts a total order on the vertices of the graph (via the o-presentation) breaking the natural symmetries in a graph. 

\section*{Acknowledgment}
I would like to thank my guide Prof. R. Ramanujam for his advice and discussions on  both technical matter and the presentation of this paper. I would also like to thank the anonymous referee for his time and for painstakingly going through this paper.

 \bibliographystyle{plain}
 \bibliography{recPred}

\begin{thebibliography}{10}

\bibitem{comon1994ordering}
Hubert Comon and Ralf Treinen.
\newblock Ordering constraints on trees.
\newblock In {\em Colloquium on Trees in Algebra and Programming}, pages 1--14.
  Springer, 1994.

\bibitem{comon1997first}
Hubert Comon and Ralf Treinen.
\newblock The first-order theory of lexicographic path orderings is
  undecidable.
\newblock {\em Theoretical Computer Science}, 176(1-2):67--87, 1997.

\bibitem{dershowitz1979orderings}
Nachum Dershowitz.
\newblock Orderings for term-rewriting systems.
\newblock In {\em Foundations of Computer Science, 1979., 20th Annual Symposium
  on}, pages 123--131. IEEE, 1979.

\bibitem{diestel2005graph}
Reinhard Diestel.
\newblock {\em Graph theory.}
\newblock Springer, 2005.

\bibitem{enderton2001mathematical}
Herbert Enderton.
\newblock {\em A mathematical introduction to logic}.
\newblock Academic press, 2001.

\bibitem{fitting2011fundamentals}
Melvin Fitting.
\newblock {\em Fundamentals of generalized recursion theory}.
\newblock Elsevier, 2011.

\bibitem{foulds2012graph}
Leslie~R Foulds.
\newblock {\em Graph theory applications}.
\newblock Springer Science \& Business Media, 2012.

\bibitem{garey1976some}
Michael~R Garey, David~S. Johnson, and Larry Stockmeyer.
\newblock Some simplified np-complete graph problems.
\newblock {\em Theoretical computer science}, 1(3):237--267, 1976.

\bibitem{gradel2007finite}
Erich Gr{\"a}del, Phokion~G Kolaitis, Leonid Libkin, Maarten Marx, Joel
  Spencer, Moshe~Y Vardi, Yde Venema, and Scott Weinstein.
\newblock {\em Finite Model Theory and its applications}.
\newblock Springer Science \& Business Media, 2007.

\bibitem{grohe2008quest}
Martin Grohe.
\newblock The quest for a logic capturing ptime.
\newblock In {\em Logic in Computer Science, 2008. LICS'08. 23rd Annual IEEE
  Symposium on}, pages 267--271. IEEE, 2008.

\bibitem{halfon2017decidability}
Simon Halfon, Philippe Schnoebelen, and Georg Zetzsche.
\newblock Decidability, complexity, and expressiveness of first-order logic
  over the subword ordering.
\newblock {\em arXiv preprint arXiv:1701.07470}, 2017.

\bibitem{hatami2011undecidability}
Hamed Hatami and Serguei Norine.
\newblock Undecidability of linear inequalities in graph homomorphism
  densities.
\newblock {\em Journal of the American Mathematical Society}, 24(2):547--565,
  2011.

\bibitem{hell2004graphs}
Pavol Hell and Jaroslav Nesetril.
\newblock {\em Graphs and homomorphisms}.
\newblock Oxford University Press, 2004.

\bibitem{immerman2012descriptive}
Neil Immerman.
\newblock {\em Descriptive complexity}.
\newblock Springer Science \& Business Media, 2012.

\bibitem{jevzek2009definability}
Jaroslav Je{\v{z}}ek and Ralph McKenzie.
\newblock Definability in substructure orderings, i: finite semilattices.
\newblock {\em Algebra universalis}, 61(1):59--75, 2009.

\bibitem{jevzek2009adefinability}
Jaroslav Je{\v{z}}ek and Ralph McKenzie.
\newblock Definability in substructure orderings, iii: finite distributive
  lattices.
\newblock {\em Algebra universalis}, 61(3-4):283--300, 2009.

\bibitem{jevzek2009bdefinability}
Jaroslav Je{\v{z}}ek and Ralph McKenzie.
\newblock Definability in substructure orderings, iv: finite lattices.
\newblock {\em Algebra universalis}, 61(3-4):301--312, 2009.

\bibitem{jevzek2010definability}
Jaroslav Je{\v{z}}ek and Ralph McKenzie.
\newblock Definability in substructure orderings, ii: finite ordered sets.
\newblock {\em Order}, 27(2):115--145, 2010.

\bibitem{karandikar2015decidability}
Prateek Karandikar and Philippe Schnoebelen.
\newblock Decidability in the logic of subsequences and supersequences.
\newblock {\em arXiv preprint arXiv:1510.03994}, 2015.

\bibitem{kudinov2009gandy}
Oleg~V Kudinov and Victor~L Selivanov.
\newblock A gandy theorem for abstract structures and applications to
  first-order definability.
\newblock In {\em Conference on Computability in Europe}, pages 290--299.
  Springer, 2009.

\bibitem{kudinov2010definability}
Oleg~V Kudinov, Victor~L Selivanov, and Lyudmila~V Yartseva.
\newblock Definability in the subword order.
\newblock In {\em Programs, Proofs, Processes}, pages 246--255. Springer, 2010.

\bibitem{kunos2015definability}
{\'A}d{\'a}m Kunos.
\newblock Definability in the embeddability ordering of finite directed graphs.
\newblock {\em Order}, 32(1):117--133, 2015.

\bibitem{kuske2006theories}
Dietrich Kuske.
\newblock Theories of orders on the set of words.
\newblock {\em RAIRO-Theoretical Informatics and Applications}, 40(01):53--74,
  2006.

\bibitem{lovasz2008graph}
L{\'a}szl{\'o} Lov{\'a}sz.
\newblock Graph homomorphisms: Open problems.
\newblock {\em manuscript available at http://www. cs. elte.
  hu/lovasz/problems. pdf}, 2008.

\bibitem{lovasz2012random}
L{\'a}szl{\'o} Lov{\'a}sz and Bal{\'a}zs Szegedy.
\newblock Random graphons and a weak positivstellensatz for graphs.
\newblock {\em Journal of Graph Theory}, 70(2):214--225, 2012.

\bibitem{matiyasevich1971diophantine}
Yuri~Vladimirovich Matiyasevich.
\newblock Diophantine representation of enumerable predicates.
\newblock {\em Izvestiya Rossiiskoi Akademii Nauk. Seriya Matematicheskaya},
  35(1):3--30, 1971.

\bibitem{quine1946concatenation}
Willard~V Quine.
\newblock Concatenation as a basis for arithmetic.
\newblock {\em The Journal of Symbolic Logic}, 11(4):105--114, 1946.

\bibitem{ramanujam2016definability}
R~Ramanujam and RS~Thinniyam.
\newblock Definability in first order theories of graph orderings.
\newblock In {\em Logical Foundations of Computer Science}, pages 331--348.
  Springer, 2016.

\bibitem{razborov20103}
Alexander~A Razborov.
\newblock On 3-hypergraphs with forbidden 4-vertex configurations.
\newblock {\em SIAM Journal on Discrete Mathematics}, 24(3):946--963, 2010.

\bibitem{robertson2004graph}
Neil Robertson and Paul~D Seymour.
\newblock Graph minors. xx. wagner's conjecture.
\newblock {\em Journal of Combinatorial Theory, Series B}, 92(2):325--357,
  2004.

\bibitem{thinniyam2017definability}
Ramanathan~S Thinniyam.
\newblock Definability of recursive predicates in the induced subgraph order.
\newblock In {\em Indian Conference on Logic and Its Applications}, pages
  211--223. Springer, 2017.

\bibitem{venkataraman1987decidability}
KN~Venkataraman.
\newblock Decidability of the purely existential fragment of the theory of term
  algebras.
\newblock {\em Journal of the ACM (JACM)}, 34(2):492--510, 1987.

\bibitem{wires2016definability}
Alexander Wires.
\newblock Definability in the substructure ordering of simple graphs.
\newblock {\em Annals of Combinatorics}, 20(1):139--176, 2016.

\end{thebibliography}
\newpage
\appendix

\section{ Defining Formulae for the Predicates in Lemma \ref{lem:oldBasicPredSubgraph}}
 \label{app:basicPredSubgraph}
	\underline{The family $\mathcal{N}$}\\
	First we need to define a few constants:
	\begin{enumerate}
		\item $\emptyset_g(x) := \forall y \; x \leq_s y$
	  	\item $N_1(x) := \emptyset_g \lessdot_s x$
	  	\item   $ N_2(x) :=  N_1 \lessdot_s x$
	  	\item $K_2(x) :=  N_2 \lessdot_s x \wedge \exists y \; x \lessdot_s y \; \wedge \; \forall z \; x \lessdot_s z \supset z=y$
	  	\item $N_3(x) :=  N_2 \lessdot_s x \; \wedge \; x \neq K_2$
  		\item $K_2N_1(x) := K_2 \lessdot_s x$
  		\item   $K_2N_2(x) := K_2N_1 \lessdot_s x \; \wedge \; N_3 \lessdot_s x $
  		\item $P_3(x) := \exists !y \; y \lessdot_s x \wedge y=K_2N_1$
  		\item $P_3N_1(x) :=  P_3 \lessdot_s x  \wedge  K_2N_2 \lessdot_s x \; \wedge
  		\forall y \; y \lessdot_s x  \supset (y=P_3 \; \vee \; y=K_2N_2)$
  		\item $S_4(x) := P_3N_1 \lessdot_s x \; \wedge \; \forall y \; y \lessdot_s x \supset y=P_3N_1$
	\end{enumerate}
	We can now get:
	\[ \mathcal{N}(x) := K_2 \nleq_s x\]
	\textbf{Remark:} We can now define $|x|=|y|$ which holds if and only if $x$ and $y$ are graphs of the same cardinality:
	\[ |x|= |y| := \forall z \; \mathcal{N}(z) \supset (z \leq_s x \iff z \leq_s y)\]
 \underline{Covers}\\
 We can define the strict order $x <_s y := x \leq_s y \; \wedge \; x \neq y$
 \[ x \lessdot_s y := x \leq_s y \; \wedge \; \forall z \; \neg (x <_s z <_s y)\]
 An edge-cover is a cover which is formed by addition of an edge:
 \[ x \lessdot_{se} y := x \lessdot_s y \; \wedge \; |x| = |y| \]
 The vertex-cover is unique and is the cover formed by adding an isolated vertex:
 \[ x \lessdot_{sv} y := x \lessdot_s y \; \wedge \; |x| \neq |y|\]

	\underline{The family $soc$}\\
	First we define the family $pac(x)$ which holds if and only if $x$ is a disjoint union of paths and cycles.
	\[ pac(x) := S_4 \nleq_s x\]
	\begin{flalign*}
	    soc(x) =& soc'(x) \wedge \forall y \: ( x \lessdot_s y \wedge pac(y)) \supset soc'(y)\\
	    &\text{where}\\
	  	soc'(x) := &x \neq \emptyset_g \; \wedge \; pac(x) \wedge \; \forall y \:( |y|_{gr} = |x|_{gr} \wedge pac(y)) \supset \neg x <_s y
	\end{flalign*}
	\underline{The family $\mathcal{T}$}\\
	First we define $forest(x)$ which holds if and only if $x$ is a disjoint union of trees
			\[ forest(x) := \forall y \; soc(y) \supset y \nleq_s x \]
	Now we can get the relation $\mathcal{T}(x)$ which holds if and only if $x$ is a tree.
		\begin{flalign*}
  			\mathcal{T}(x) :=& forest(x) \; \wedge \; \forall y \; (forest(y) \; \wedge \; |x|_{gr} = |y|_{gr}) \supset \neg x <_s y
		\end{flalign*}	
	\underline{The predicate $conn(x)$}
	\[ conn(x) := \exists y \; \mathcal{T}(y) \; \wedge \; |y|=|x| \; \wedge \; y \leq_s x\]

	\underline{The families $\mathcal{C}, \mathcal{K}, \mathcal{S}, \mathcal{P}$}\\
	\[ \mathcal{C}(x) := soc(x) \; \wedge \; conn(x) \]
	$\mathcal{K}(x)$  holds if and only if $x$ is a clique:
	\[ \mathcal{K}(x) := \forall y \neg(x \lessdot_{se} y)\]
	
	$\mathcal{P}(x)$  holds if and only if $x$ is a path:
	\[ \mathcal{P}(x) := \mathcal{T}(x) \; \wedge \; S_4 \nleq_s x\]
	Since paths are totally ordered by $\leq_s$ (equivalently by cardinality), by Observation \ref{obs:totalOrderDefble} all of them are definable as constants.\\
	$\mathcal{S}(x)$  holds if and only if $x$ is a star:
	\[ \mathcal{S}(x) := \mathcal{T}(x) \; \wedge \; P_4 \nleq_s x\]
	Similar to paths, all stars are definable as constants.

 \noindent \textbf{Definability of the constant $double3star$}
 \[ double3star(x) := \mathcal{T}(x) \; \wedge \; |x|=6 \; \wedge \; P_4 \leq_s x \; \wedge \; P_5 \nleq_s x \; \wedge \; S_5 \nleq_s x \]

\end{document}